\documentclass[11pt]{article}
\pdfoutput=1
\usepackage{fullpage}
\usepackage[dvipsnames,usenames]{xcolor}
\usepackage{amsthm,amssymb,amsmath,amsfonts}  
\usepackage{xspace,enumerate}
\usepackage{enumitem}
\usepackage{thmtools}
\usepackage{thm-restate}

\usepackage{graphicx}
\usepackage{graphics}
\usepackage[noadjust]{cite}
\usepackage{comment}
\usepackage{todonotes}
\usepackage{verbatim}
\usepackage{units}
\usepackage{color}
\usepackage{epsfig}
\usepackage{marvosym}

\usepackage{authblk}

\usepackage[hypertexnames=false,colorlinks=true,urlcolor=Blue,citecolor=Green,linkcolor=BrickRed]{hyperref}
\usepackage[capitalise]{cleveref}

\def\eps{\varepsilon}

\theoremstyle{plain}
\newtheorem{theorem}{Theorem}
\newtheorem{lemma}[theorem]{Lemma} 
\newtheorem{corollary}[theorem]{Corollary}  
\newtheorem{proposition}[theorem]{Proposition}  
\newtheorem{fact}[theorem]{Fact}
\newtheorem{observation}[theorem]{Observation}
\newtheorem{definition}[theorem]{Definition}

\newtheorem{claim}[theorem]{Claim} 
\newtheorem{conjecture}{Conjecture}
\newtheorem{question}{Question}

\author[1]{Amir Abboud\footnote{This work is part of the project CONJEXITY that has received funding from the European Research Council (ERC) under the European Union's Horizon Europe research and innovation programme (grant agreement No.~101078482). Additionally, it is supported by an Alon scholarship and a research grant from the Center for New Scientists at the Weizmann Institute of Science. }}
\author[1]{Nathan Wallheimer}
\affil[1]{
Weizmann Institute of Science\\
\href{mailto:amir.abboud@weizmann.ac.il}{amir.abboud@weizmann.ac.il}, \href{mailto:nathan.wallheimer@weizmann.ac.il}{nathan.wallheimer@weizmann.ac.il}}

\title{Worst-Case to Expander-Case Reductions: \\ Derandomized and Generalized} 

\begin{document}

\maketitle

\begin{abstract}
A recent paper by Abboud and Wallheimer [ITCS 2023] presents self-reductions for various fundamental graph problems, 
which transform worst-case instances to expanders, thus proving that the complexity remains unchanged if the input is assumed to be an expander. 
An interesting corollary of their self-reductions is that if some problem admits such reduction, then the popular algorithmic paradigm based on expander-decompositions is useless against it. 
In this paper, we improve their core gadget, which augments a graph to make it an expander while retaining its important structure. Our new core construction has the benefit of being simple to analyze and generalize while obtaining the following results: 
\begin{itemize}
\item A derandomization of the self-reductions, showing that the equivalence between worst-case and expander-case holds even for deterministic algorithms, 
and ruling out the use of expander-decompositions as a derandomization tool. 
\item An extension of the results to other models of computation, such as the Fully Dynamic model 
and the Congested Clique model. In the former, we either improve or provide an alternative approach 
to some recent hardness results for dynamic expander graphs by Henzinger, Paz, and Sricharan [ESA 2022]. 
\end{itemize}
In addition, we continue this line of research by designing new self-reductions for more problems, such as Max-Cut and dynamic Densest Subgraph, and demonstrating that the core gadget 
can be utilized to lift lower bounds based on the OMv Conjecture to expanders. 
\end{abstract}

\thispagestyle{empty}
\newpage
\setcounter{page}{1}

\section{Introduction}
\label{sec:introduction}
When studying the complexity of any graph problem, it is natural to ask whether the problem can be solved faster on \emph{expanders}, i.e., random-like, well-connected graphs that satisfy a certain definition of expansion. 

\begin{question}
\label{q:question1}
Are expanders worst-case instances of my problem?\end{question}

The motivation for such a question comes from multiple sources.
First, it is inherently interesting to understand how the rich mathematical structure of expanders affects the complexity of fundamental problems such as shortest paths, cuts, matchings, subgraph detection, and so on. After all, expanders are among the most important graph families in computer science.
Second, expanders exhibit some of the most algorithmically useful properties of \emph{uniformly random} graphs, and so this question may help understand the \emph{average-case} complexity. 
Third, graphs that arise in applications may be expanders (e.g., in network architecture). 
Moreover last but not least is the hope that if we solve a problem faster on expanders, we will also be able to solve it in the worst case by utilizing the popular \emph{expander decomposition method}, which we discuss soon. 
Unless explicitly stated otherwise, we use the conductance-based notion of $\phi$-expanders, whose precise definition can be found in \cref{sec:preliminaries}, and we say that a graph is an expander if it is an $\Omega(1)$-expander.

It is possible to cook up problems for which the answer to \cref{q:question1} is negative. For instance, we can solve connectivity in constant time if the input is promised to be an expander, but it requires linear time in the worst case. A less obvious example with such gaps is counting spanning trees~\cite{li2023new}. 
However, for many (perhaps most) interesting graph problems, the answer seems to be \emph{positive}: expanders do not make the problem any easier. 
In other words, the \emph{expander-case} is also worst-case. But how do we prove that? Let us discuss three methods and their drawbacks. 
\begin{enumerate}
\item The first and most obvious method is to prove a lower bound for the problem on expander instances that matches the worst-case upper bound. 
Technically, this may follow directly from the existing lower bounds for the problem since they are often proved on random-like graphs (e.g., for distance oracles \cite{sommer2009distance}), 
or it may require some modifications to the lower bound proofs (e.g., for dynamic graph problems \cite{henzinger2022fine}).

The main drawback with this approach is that we are interested in answering \cref{q:question1} even when (or rather, especially when) we have not already resolved the worst-case 
time complexity of our problem, in which case we do not even have a matching lower bound (e.g., the Maximum Matching problem). 
Another drawback is that when asking for the \emph{fine-grained complexity} of problems, the existing lower bounds are usually conditioned on strong assumptions, 
and one may hope to get an unconditional answer to \cref{q:question1}. 

The next two approaches resolve these drawbacks since they are based on \emph{worst-case to expander-case self-reductions} (WTERs). 
Such techniques show equivalence between the expander-case complexity and the worst-case complexity. 

\item The second approach uses the \emph{expander decomposition method}.
This is a popular paradigm in recent years that suggests we can solve graph problems by (1) decomposing the
graph into vertex-disjoint expanders with a small number of edges between them, (2) solving the problem on each expander separately, and (3) combining all the answers efficiently. Step (3) requires problem-specific techniques.
If we can solve steps (1) and (3) for a problem, then we have effectively shown that any improvement on expanders will yield an improvement on worst-case graphs, giving a positive answer to \cref{q:question1}.
Efficient algorithms for computing such expander decompositions (for step (1)) are known both in the Word-RAM ~\cite{Kaplan22,saranurak2019expander,kannan2004clusterings,spielman2004nearly,orecchia2008partitioning} and in other models of computation including dynamic~\cite{saranurak2019expander}, distributed~\cite{chang2019improved}, and recently even in streaming~\cite{filtser2022expander}.
Applications of this paradigm have led to many breakthroughs in recent years to problems such as Maximum Flow~\cite{chen2022maximum}, Dynamic Connectivity~\cite{goranci2021expander}, Gomory-Hu Trees~\cite{AKT21}, Minimum Spanning Trees~\cite{NSW17}, and 
Triangle Enumeration~\cite{chang2019improved}.

Self-reductions of this form are called \emph{ED-WTER}s. 
A recent paper by Abboud and Wallheimer~\cite{abboud2023worst} proposed an alternative, simpler method of self-reductions that \emph{do not} use expander decompositions.
They call this method \emph{Direct-WTER}s, which we discuss next, and it is not only simpler but also yields stronger qualitative and quantitative results.

\item The third and most direct method is to show that any graph can be turned into an expander without affecting the solution to the problem or increasing the size of the graph by too much. 
In their paper, Abboud and Wallheimer~\cite{abboud2023worst} gave the following definition, which we slightly reframe to fit our discussion more accurately: 
\begin{definition}[Direct-WTER{~\cite[Definition 2]{abboud2023worst}}]
\label{def:direct}
A direct worst-case to expander-case self-reduction to a graph problem $\mathcal{A}$, is an algorithm that given any instance $G$ with $n$ vertices and $m$ edges, 
computes in $\tilde{O}(n+m)$ time a graph $G_{exp} := (V_{exp},E_{exp})$ with the following guarantees:
\begin{itemize}
\item $G_{exp}$ is an $\Omega(1)$-expander with high probability. 
\item The \emph{blowup} in the number of vertices and edges in $G_{exp}$ is $|V_{exp}| \leq K$ and $|E_{exp}| \leq M$ for some $K := K(n,m)$ and $M := M(n,m)$. 
\item The solution $\mathcal{A}(G)$ can be computed from the solution $\mathcal{A}(G_{exp})$ in $\tilde{O}(m+n)$ time. 
\end{itemize}
\end{definition}

\noindent
Direct-WTERs can be used to show equivalence between the complexity of polynomial-time problems and their complexity on $\Omega(1)$-expanders. 
Namely, if problem $\mathcal{A}$ is a polynomial-time problem that admits a Direct-WTER, then $\Omega(1)$-expanders are worst-case instances of $\mathcal{A}$ (ignoring poly-logarithmic factors).

The main contribution of Abboud and Wallheimer was to show that some fundamental problems, such as $k$-Clique Detection and Maximum Matching, admit 
 simple Direct-WTERs. In particular, their Direct-WTERs make a graph an expander by employing a \emph{core gadget} that augments it with $O(n)$ vertices and $O(m + n\log n)$ random edges 
and then applies additional gadgets that control the solution. In particular, they obtain a near-linear blowup. 
Their results are surprising because such Direct-WTERs do not employ any of the heavy machinery that usually comes with expander decompositions, 
yet they output quantitatively better expanders: the outputs of Direct-WTERs are $\Omega(1)$-expanders (by definition), whereas expander decompositions can only produce $O(1/\log n)$-expanders (that are not as expanding)~\cite{saranurak2019expander,alev2017graph}. 

%Let us summarize the advantages of Direct-WTERs over ED-WTERs. 
%First, they are quantitatively better: they prove that the worst-case time complexity is the same (up to log factors) as that on $\Omega(1)$-expanders, 
%whereas expander decompositions are limited to output $O(1/\log n)$-expanders (which are worse) in some cases~\cite{SW19,alev2017graph}. 
%Second, they are easier to design because step (3) in the expander decomposition method can be quite involved. 
%And third, they do not rely on the ``heavy'' algorithmic machinery that comes with expander decompositions.
%This makes the positive answers we get for Question~\ref{q:question1} more satisfying.

Furthermore, the simplicity of such Direct-WTERs leads to interesting and important messages to algorithm designers, as observed in~\cite{abboud2023worst}: 
The expander decomposition method is useless in the presence of Direct-WTERs because decomposing a graph into $o(1)$-expanders is meaningless when we can 
assume that the input graph is already an $\Omega(1)$-expander after a simple modification. 
This addressed (with a negative answer) a question that many researchers have wondered about as they looked for the next breakthrough to be obtained via the expander decomposition method: 
\begin{question} 
\label{q:question2} 
Are expander decompositions the key to solving my problem? 
\end{question}
\end{enumerate}

\paragraph{This work.} Motivated by the appeal of the method of Direct-WTERs towards answering \cref{q:question1} and \cref{q:question2}, our goal is to develop this theory further. 
Toward that, we address the two main limitations that were highlighted in~\cite{abboud2023worst}: (1) the Direct-WTERs are \emph{randomized} whereas ED-WTERs are deterministic~\cite{chuzhoy2020deterministic}, 
and (2) the results are restricted to the Word-RAM model, whereas expander decompositions are popular tools in other models as well. 
In addition, we continue their line of work by providing Direct-WTERs to additional problems. 
Let us motivate these two topics before stating our results formally. 

\subsection{Deterministic Direct-WTERs}
\label{sub:derandomization}

The randomized Direct-WTERs of Abboud and Wallheimer~\cite{abboud2023worst} prove that $\Omega(1)$-expanders are worst-case instances of many problems \emph{if} we allow algorithms to be randomized. 
They leave us wondering if perhaps $\Omega(1)$-expanders are truly easier for deterministic algorithms.
We remark that while it is believed that all algorithms can be derandomized by incurring a small polynomial blowup (as in $P=BPP$), it is far from clear that this blowup can be made $n^{o(1)}$ 
(see \cite{chen2023guest} for the state-of-the-art on such results). 
Can we provide a positive answer to \cref{q:question1} with respect to \emph{deterministic} algorithms by designing \emph{deterministic} Direct-WTERs?

Additional motivation comes from the hope of using the expander decomposition method in order to get breakthrough \emph{derandomization} results, along the lines of \cref{q:question2}, for problems where the current randomized algorithms are much faster than the current deterministic algorithms.
Indeed, deterministic expander decompositions~\cite{chuzhoy2020deterministic} have already played a major role in some of the most remarkable derandomization results of recent years, e.g. for Global Min-Cut \cite{KT19,saranurak2021simple,LP20,li2021deterministic,henzinger2024deterministic}. 
\begin{question}
\label{q:question3}
Are expander decompositions the key to derandomizing my algorithm?
\end{question}
\noindent
Similarly to the observation in~\cite{abboud2023worst}, \emph{deterministic Direct-WTER} also convey a message to algorithm designers: that the answer to the above question is negative, i.e., 
expander decompositions are useless for derandomizing the problem. 

Motivated by this, our first result is a derandomization of the \emph{core gadget} in~\cite{abboud2023worst}, 
resulting in deterministic Direct-WTERs for various problems. In particular, we show that 
all problems admitting \emph{randomized} Direct-WTERs in~\cite{abboud2023worst}, and some additional problems, such as the Max-Cut problem, admit deterministic direct-WTERs.
\begin{theorem}
\label{thm:oldproblems}
The following problems admit deterministic Direct-WTERs: 
Maximum Matching, Minimum Vertex Cover, $k$-Clique Detection, $k$-Clique Counting, 
Max-Clique, Max-Cut, Minimum Dominating Set, and $H$-Subgraph Detection ($m = \tilde{O}(n)$ and $H$ does not contain pendant vertices). 
\end{theorem}
\noindent
We provide formal definitions and an overview of all problems mentioned below in \cref{app:problems}. 

An important feature of our deterministic core gadget is that it remains simple and efficient. Interestingly, this stands in contrast with other derandomization results 
in fine-grained complexity, which often tends to involve sophisticated methods and some slowdown (see, e.g., ~\cite{chan2020reducing,fischer2023deterministic}). 
As discussed earlier, simplicity is an important aspect of Direct-WTERs, not just because it makes them more accessible to the community, 
but also because it strengthens the message that expander decompositions become useless in the presence of Direct-WTERs. 
In \cref{sec:overview}, we provide an intuitive overview of how our derandomization is obtained by a modification to the core gadget of Abboud and Wallheimer~\cite{abboud2023worst} 
and in \cref{sec:coregadget} we provide the construction itself. 

Interestingly, our Direct-WTERs also improve upon the blowup in the number of added edges over the Direct-WTERs in~\cite{abboud2023worst}, resulting in $\Omega(1)$-expanders with $O(n)$ vertices and $O(m+n)$ edges, whereas in~\cite{abboud2023worst}, 
the expanders have $O(m + n\log n)$ edges. Note that this blowup is optimal since any expander is connected and, therefore, must contain $\Omega(n)$ edges. 
We also demonstrate that our core gadget preserves the following graph properties: 
  (1) Bipartiteness-preserving; we can modify the core gadget so that if $G$ is bipartite, then so is the expander, and (2) Degree-preserving; if the maximum degree in $G$ is $\Delta$, then the maximum-degree in $G_{exp}$ is $2\Delta + O(1)$. 

\paragraph{Remarks:}
\begin{enumerate}

\item For the exponential-time problems in \cref{thm:oldproblems}: Minimum Vertex Cover, Minimum Dominating Set, Max-Clique, and Max-Cut, 
the blowup in the number of vertices in the output graph must be subject to stronger restrictions than for polynomial-time problems. In particular, the blowup should be $n+o(n)$, 
to show that expanders are worst-case instances.\footnote{
Otherwise, an exponential speed-up on expanders does not necessarily translate to an exponential speed-up on general graphs.}
For such problems, we employ a generalized core gadget, which gives a tradeoff between the conductance and the blowup, resulting in Direct-WTERs providing, for every $0 < \eps \leq 1$,  
conductance $\Omega(\eps)$ and blowup $\eps n$. 

\item For $k$-Clique Detection, we slightly improve the parameters over the Direct-WTERs given in~\cite{abboud2023worst} for this problem.
In~\cite{abboud2023worst}, the output is an $\Omega(1/k^2)$-expander $G_{exp}$, with a blowup of $\Theta(nk)$ vertices and $\Theta(k^2m)$ edges, where each $k$-clique in $G$ corresponds to 
$k!$ $k$-cliques in $G_{exp}$. Our improved Direct-WTER produces an $\Omega(1)$-expander $G_{exp}$ with $O(n)$ vertices and $O(m+n)$ edges (even if $k = \omega(1)$), such that every $k$-clique in $G$ corresponds to $k+1$ $k$-cliques in $G_{exp}$. This enhancement makes our reduction more suitable for parameterized algorithms and larger (non-constant) values of $k$. 

\item For the Max-Cut problem, there is an interesting related work on the approximation variant of the problem on expanders. 
A famous algorithm by Goemans and Williamson~\cite{goemans1995improved} obtains a $>0.878$-approximation for the maximum cut in general graphs, based on a Semidefinite Programming relaxation. 
In search of other, perhaps simpler methods for approximating the max-cut beyond the trivially obtained $1/2$-approximation\footnote{Which follows from the fact that the maximum cut is at least $m/2$ in any graph with $m$ edges.}, Trevisan~\cite{trevisan2009max} posed the following question: Is there a combinatorial algorithm that achieves better than  $1/2$-approximation? 

A positive answer to this question was given by Kale and Seshadhri~\cite{kale2010combinatorial}, which remains the current-best combinatorial algorithm for this problem.
A recent paper by Peng and Yoshida~\cite{peng2023sublinear} addresses this question on expanders, providing a combinatorial algorithm for approximating the maximum cut on $\phi$-expanders. 
Namely, the authors provide an algorithm that given $\eps$, computes a $(1/2+\eps)$-approximation, subject to $\eps = O( \phi^2 )$. In more detail, it computes a value $x$ such that $(1/2+\eps)MC(G) \leq x \leq MC(G)$, where $MC(\cdot)$ denotes the cardinality of the maximum cut in $G$. Moreover, its running time is sublinear when $\phi$ is a constant. 
Our Direct-WTER for Max-Cut in \cref{thm:oldproblems}, on the other hand, is a reduction that given $G$ and $0<\phi \leq 1$, 
outputs an $\Omega(\phi)$-expander $G_{exp}$, and the maximum cut in $G_{exp}$ is $MC(G_{exp}) \leq (1+4\phi)MC(G)$, assuming the input graph is not too sparse, say, $m = \omega(n)$. 

Can we apply the Direct-WTER and then use the algorithm of Peng and Yoshida to get a combinatorial algorithm that $(1/2+\eps)$-approximates the maximum cut in general graphs? 
Perhaps we may even improve upon the algorithm by Kale and Seshadhri, as both the Direct-WTER and the algorithm of Peng and Yoshida are very efficient. However, the approximate value we get from this approach is $(1/2+\eps)MC(G) \leq x  \leq (1+4\phi)MC(G)$, 
or equivalently $(1/2+\eps)/(1+4\phi) MC(G) \leq x \leq MC(G)$. 
For this approximation ratio to be larger than $1/2$, we need to pick $\eps > 2\phi$, but recall the constraint $\eps \leq O(\phi^2)$. Hence, this approach fails. 
This is not so surprising, since our Direct-WTER is quite elementary so we do not expect to make use of it as a subroutine inside another algorithm. 
Instead, this result should be interpreted as a limitation to algorithms for Max-Cut on $\phi$-expanders; that one cannot obtain a $(1/2+\eps)$-approximation for some $2\phi<\eps \leq 1/2$ 
unless this (unlikely) approach works. 
\noindent

\end{enumerate}

\subsection{Direct-WTERs in the Fully Dynamic setting}
\label{sec:models}
Before this work, Direct-WTERs were limited to the (randomized) Word-RAM model, whereas ED-WTERs could address many other models. 
One particular area in which expanders, expander decomposition, and derandomization are important subjects is the area of dynamic graph algorithms. 
Let us focus on the Fully Dynamic model of computation, where the goal is to maintain the solution of a problem in a graph undergoing edge updates, i.e., edge insertions and deletions. 
Dedicated tools have been developed for maintaining an expander decomposition in this model and subsequently achieved major breakthroughs (e.g., \cite{NSW17}).
Notably, derandomization is a central concern in the dynamic setting because deterministic algorithms are essentially the only ones that work against adaptive adversaries (see, e.g., \cite{beimel2022dynamic}). 
For these reasons, the three main questions outlined above are particularly interesting in this model. Our first question, in this context, is whether Direct-WTERs can be adapted to the 
Fully Dynamic setting. 

An important related work is a recent paper by Henzinger, Paz, and Sricharan~\cite{henzinger2022fine} (abbreviated as HPS), 
who initiated the study of \cref{q:question1} in the dynamic model, regarding the complexity of fundamental problems on \emph{dynamic expanders}. 
A dynamic expander is a dynamic graph that undergoes edge updates but remains an $\Omega(1)$-expander 
at any point in time. They adapt lower bound proofs from fine-grained complexity, so that they hold even on dynamic expanders. 
Their techniques differ from the self-reduction approach of Direct-WTERs and correspond to the first of the three methods outlined above to answer \cref{q:question1}. 
In particular, they base their results on the \emph{Online Matrix Vector} (OMv) conjecture, which was introduced by Henzinger et al.~\cite{henzinger2015unifying} 
to prove lower bounds for various dynamic problems. 
HPS obtained their results by adapting these lower bound proofs to the case of constant-degree expanders, 
proving that dynamic expanders whose maximum degree remains bounded by a constant are OMv-hard. The problems they consider in this context are Maximum Matching, Densest Subgraph, and $st$-Shortest Path (abbreviated as $st$-SP). 
This leaves us wondering with the following questions: 
\begin{enumerate} 
\item Are expanders in higher density regimes, whose maximum degree is not bounded by a constant, also OMv-hard instances? 
\item Can the techniques that are used in the static setting to prove \cref{thm:oldproblems} contribute to this study by, e.g., providing a simpler, 
or alternative method to prove that certain problems remain OMv-hard on expanders? 
\end{enumerate} 
\noindent 

Our second main result is an adaptation of the deterministic core gadget to the dynamic setting, resulting in a deterministic, dynamic algorithm for maintaining a dynamic $\Omega(1)$-expander, 
whose running time is amortized $\tilde{O}(1)$ per edge update. Subsequently, we show \emph{Dynamic Direct-WTERs} (abbreviated as DD-WTERs) to various problems,\footnote{Namely, to all problems in the previous theorem, except that we do not discuss the exponential-time problems in the dynamic model in this work.} 
thus proving that $\Omega(1)$-expanders are worst-case instances and that the expander decomposition method is useless against them. We present and discuss the formal definition of DD-WTERs in \cref{app:ddwter}.
\begin{theorem}
\label{thm:dynamic}
The following problems admit a DD-WTER: 
Maximum Matching, Bipartite Perfect Matching, Densest Subgraph (in graphs with $m > 42n$ edges), $k$-Clique Detection, $k$-Clique Counting, and $H$-Subgraph Detection (where $m = \tilde{O}(n)$ and $H$ does not contain pendant vertices). 
\end{theorem}
\noindent 

In addition, our results also have some interesting implications related to the work of HPS and the questions above. 
\begin{enumerate}
\item For the Densest Subgraph problem, Henzinger et al.~\cite{henzinger2015unifying} prove a lower bound of $n^{1/3-o(1)}$ per update under OMv for \emph{general graphs}. By modifying their reduction, HPS were able to show a weaker $n^{1/4-o(1)}$ lower bound for \emph{expanders}.\footnote{
We remark that the proof in Henzinger et al.~\cite[Corollary 3.26]{henzinger2015unifying} gives an $n^{1/3-o(1)}$ lower bound under OMv, while the introductions of both~\cite{henzinger2015unifying} and HPS~\cite{henzinger2022fine} mention, erroneously, an $n^{1/2-o(1)}$ lower bound.}

As a consequence of our DD-WTER, we conclude that the $n^{1/3-o(1)}$ lower bound for general graphs also holds for expanders. One subtlety towards this result is that the reduction of Henzinger et al. produces very sparse graphs with $m\leq 2n$ edges while our DD-WTER assumes that $m> 42n$. In \cref{app:densest}, we discuss how to modify the original reduction so that denser graphs are produced. 
Another strength of the DD-WTER compared to HPS is that it is a self-reduction, hence it does not depend on the OMv Conjecture to get a lower bound (at least when the $m > 42n$ assumption is made)\footnote{We remark that the limitation of $m > 42n$ in this approach is due to the fact that to augment a graph to become an expander requires at least adding some amount of edges, which unavoidably affects the densest subgraph.}. 

\item For the Maximum Matching problem on graphs of maximum degree $O(n^t)$, for any $0 \leq t \leq 1$, we 
show that this problem is OMv-hard on $\Omega(1)$-expanders as well. This improves upon HPS, who only prove it for the constant-degree case (i.e. $t = 0$)~\cite[Theorem 12]{henzinger2022fine}, while implicitly leaving an open question regarding $t>0$. While intuitively, the constant-degree case should be the most difficult to prove a lower bound for, it is not immediate, as techniques that artificially increase the degrees in the graph (e.g. attaching a star to every vertex) tend to also increase the number of vertices, thus resulting in weaker lower bounds. Instead, our result is obtained by combining a lower bound by HPS for graphs (not expanders) of maximum degree $O(n^t)$~\cite[Theorem 12]{henzinger2022fine}, with our degree-preserving DD-WTER for Maximum Matching. Hence, we get the following corollary:
\begin{corollary}
For any $0\leq \eps,t\leq 1$ and any constant $\eps >0$, there is no dynamic algorithm for maintaining a maximum matching on $\Omega(1)$-expanders with maximum degree $O(n^{t})$, 
with amortized $O(n^{(1+t)/2-\eps})$ update time and $O(n^{1+t-\eps})$ query time, unless the OMv Conjecture is false.
\end{corollary}

\item We demonstrate that our core gadget is useful even outside the context of self-reductions by showing a DD-WTER for \emph{Graphical OMv}, an equivalent graph formulation of the OMv problem. 
This implies that we can lift OMv-based lower bounds to $\Omega(1)$-expanders for many problems, in particular, for problems such as Maximum Matching and $st$-SP, which HPS considered. 
We demonstrate the power of this technique by proving that $st$-SP is OMv-hard on $\Omega(1)$-expanders.\footnote{Note that in comparison with HPS, 
we do not prove hardness for graphs of constant degree, which is outside the scope of our paper.} 

\begin{proposition}[$st$-SP is OMv-hard on $\Omega(1)$-expanders]
\label{prop:st-SP}
For any $\eps > 0$, there is no dynamic algorithm for the dynamic $st$-SP problem on $\Omega(1)$-expanders, with polynomial preprocessing time, 
$O(m^{1/2-\eps})$ update time, and $O(m^{1-\eps})$ query time, assuming the OMv Conjecture. 
\end{proposition}
\end{enumerate}

Finally, let us emphasize another aspect of how our work and HPS differ. One of the main motivations for our work is to address \cref{q:question2} and \cref{q:question3} about the 
applicability of expander decompositions in algorithms, whereas HPS main motivation is to gain a better understanding of the complexity of dynamic problems on various graph families, 
including expanders, along the lines of \cref{q:question1}. 
As we have seen, DD-WTERs imply that expander decompositions are useless because they show equivalence between the worst-case and the expander-case complexities. 
We remark that the results of HPS imply the same, assuming their obtained lower bounds are tight. Hence, it implies a \emph{conditional} answer to these questions, 
whereas self-reductions imply an \emph{unconditional} one. 

\subsection{Expanders in distributed models.}
The above results demonstrate that the impact of Direct-WTERs goes beyond the classic Word-RAM model of computation. A natural continuation of this is to apply these techniques to more models of computation, where expanders, expander decompositions, and derandomization are important subjects. 
In \cref{sec:distributed}, we discuss the applicability of Direct-WTERs in distributed models of computation. We show that while there are limitations in adapting Direct-WTERs to CONGEST, it is possible to do so in CONGESTED-CLIQUE and MPC (Massively Parallel Computation).

\paragraph{Roadmap.} A technical overview is given in Section \ref{sec:overview}, where we also explain the differences compared to~\cite{abboud2023worst}. 
Then, after some preliminaries in Section \ref{sec:preliminaries}, we provide the full details of our derandomized core gadget in Section \ref{sec:coregadget}. 
Section~\ref{sec:variants} presents variants of the core gadget, including a fully dynamic adaptation. 
Section~\ref{sec:wters} presents deterministic Direct-WTERs for Max-Cut, Densest Subgraph, and Graphical OMv. In addition, Section \ref{sec:oldwters} presents a dernadomization and dynamization of all the 
Direct-WTERs that appeared in~\cite{abboud2023worst}, and additional Direct-WTERs for related problems. 
Section \ref{sec:distributed} addresses Direct-WTERs in distributed models of computation. 
Appendix \ref{app:ddwter} contains a formal definition of DD-WTERs.

\section{Technical Overview}
\label{sec:overview}
In this section, we summarize the core gadget of Abboud and Wallheimer~\cite{abboud2023worst} (henceforth, AW) that is used in all their Direct-WTERs, 
and then present our modification. Roughly, their construction boils down to the following procedure.

\paragraph{AW's core gadget.}
 Given a graph $G = (V,E)$, add a set $U$ of $n$ vertices called \emph{expansion layer}. Then, for every vertex $v\in V$, take a sample of $\deg_G(v) + O(\log n)$ vertices in $U$ and make them neighbors of $v$. 
 Clearly, the size of this graph and the running time are both $O(m + n\log n)$. In addition, they prove that it is an $\Omega(1)$-expander with high probability by showing 
that the probability that a cut in $G_{exp}$ ends up being sparse is very low. 
However, there is no guarantee that the problem's solution to $G$ can be computed easily from the solutions to $G_{exp}$. To this end, AW provide additional, mostly simple gadgets to control the solution in $G_{exp}$ in a predictable manner while preserving the conductance of $G_{exp}$ up to a constant factor. 

\paragraph{Derandomizing the core gadget.}
The first limitation of AW's approach that we aim to resolve in this paper is their use of randomness.
To do so, one could attempt to introduce an \emph{explicit}, $d$-regular bipartite $\Omega(1)$-expander $X$ between $V$ and $U$. However, which degree $d$ should we choose? 
For starters, we consider picking some constant $d \geq 3 $. However, this approach might fail if $G$ contains cuts with many internal edges and few out-going edges. 
For instance, if $G$ includes a cut $S$ of $\sqrt{n}$ vertices that induce an isolated $\sqrt{n}$-clique, then in $G_{exp}$ we have $vol(S) = \Theta(n)$, 
while the number of out-going edges added by the expander is $O(|S|d) = O(\sqrt{n})$, thus $\phi(S) = O(1/\sqrt{n})$. 
On the opposite end, we could select a complete bipartite graph with $\Theta(n^2)$ edges. While this will address the problem and indeed result in an $\Omega(1)$-expander, 
this approach is undesirable due to the excessive number of added edges. 
Hence, the most natural strategy to consider is incorporating a $d$-regular expander for $d = \lceil \frac{m}{n} \rceil$. 
To be more precise, since explicit constructions exist only for $d \geq 3$, the approach is to pick $d = \lceil \frac{m}{n} \rceil + 3$ (as it could be the case that, for example, $m=0$). While the number of added edges is now $O(m+n)$, this approach is slightly na\"ive, as it does not even address the above problem of isolated $\sqrt{n}$-clique in a graph with $O(n)$ edges. The crux of the issue lies in the fact that $G$ is not regular. If $G$ were $d$-regular, then the inclusion of a $d$-regular expander would resolve our problem. 
One can look for expanders with a certain degree sequence, matching the one of $ G $, but constructing those seems challenging. 

The starting point for our modification is the observation that the randomness in AW's constructions is coming to accomplish \emph{two things at once}: first and foremost, 
choosing random edges breaks structure and creates a  ``random like'' graph. However, second, and less obviously, it is a way to achieve a certain ``balanced allocation'' of neighbors in $U$, resulting in nearly uniform degrees in $U$, which is crucial for their analysis to work. 
This leads to our modification, in which we substitute their single expansion layer $U$ with two layers: one for balanced allocation denoted $L$, connected to $V$ using a load-balancing algorithm, ensuring that the degrees in this layer are about $\frac{2m}{n}$, and another layer for expansion denoted $R$, connected to $L$ using the edges of a bipartite, $(\frac{2m}{n}+3)$-regular $\Omega(1)$-expander. 
A random graph can accomplish the construction of the first layer, as AW did, but it could also be accomplished \emph{in any other way}, such as the standard \emph{Round-Robin} algorithm without any randomness.
Then, all we have to do is use an explicit construction (described in Appendix \ref{app:explicit}) of a \emph{regular} $\Omega(1)$-expander. 
This is the whole deterministic core gadget in Section \ref{sec:coregadget}. Using it, we follow the steps of AW and add various gadgets to $G_{exp}$, such that the solution for $G$ can be retrieved from the solution for $G_{exp}$. 

\paragraph{Dynamizing the core gadget.}
The above core gadget is not suitable for the fully dynamic setting, where $G$ undergoes edge insertions and deletions. The main issue arising is that when the degree of a vertex $v \in V$ increases due to edge insertions, we need to allocate to $v$ additional neighbors from $L$ while preserving approximately balanced degrees in $L$. 
One attempt to solve this issue is to add an edge from $v$ to a minimum-degree vertex in $L$, which can be computed quickly using a priority queue. 
While this ensures balanced degrees in $L$, there is a subtle issue that the minimum-degree vertex in $L$ might already be a neighbor of $v$; hence we can not add another edge to it because that would create a parallel edge (which we aim to avoid). 
Therefore, we slightly modify this heuristic, and instead, our suggested approach is compute the successor of the minimum-degree vertex in $L$ repeatedly, 
until a minimum-degree vertex in $L \setminus N(v)$ is reached, and then make it a neighbor of $v$. 
The required computation here is proportional to the degree of $v$, which might not be constant, 
but using lazy updates this approach results in amortized cost $O(1)$. However, a possible issue arising in this algorithm is that the minimum-degree vertex in $L \setminus N(v)$ is 
not necessarily a minimum-degree vertex in $L$, so it needs to be clarified that the degrees in $L$ remain balanced, as otherwise the graph might not be an $\Omega(1)$-expander. 
Nonetheless, we prove that the degrees in $L$ become overly imbalanced only after $\Omega(m+n)$ edge insertions. At this stage, there is enough credit to reconstruct $G_{exp}$ from scratch, resulting in a total amortized cost of $O(1)$ per edge insertion. 
In general, edge deletions do not require further computation. However, too many edge deletions or insertions will make $X$ too dense or too sparse compared to $G$, resulting either in a large blowup or small conductance. This may happen only after $\Omega(m+n)$ updates, thus it can also be handled in amortized cost $O(1)$ by periodic reconstruction.

\section{Preliminaries}
\label{sec:preliminaries}
Let $G = (V,E)$ be an undirected simple graph. Throughout the paper, we use $n$ and $m$ to denote the number of edges in vertices in $G$, respectively. 
Denote the neighborhood of $v$ by $N(v) := \{ u \in V \mid uv \in E \}$ and the degree of $v$ by $\deg(v)$. 
We say that $G$ is $d$-regular if $\deg(v) = d$ for all $v \in V$. 
A vertex $v$ is called a \emph{pendant vertex} if $\deg(v) = 1$. 
The set of all edges with one endpoint in $S \subseteq V$ and another endpoint in $T \subseteq V$ is denoted by
$E(S,T) := \{ uv \in E \mid u \in S, v \in T \} $, and its cardinality is denoted by $e(S,T) := |E(S,T)|$. We call $E(S,V \setminus S)$ the \emph{out-going edges} of $S$. 
We employ subscripts to indicate which graphs we refer to when it is not clear from the context. 
For instance, $\deg_H(v)$ denotes the degree of vertex $v$ in a graph $H$.

\paragraph{Conductance and edge-expansion.} 
Let $S \subseteq V$ be a cut. 
The \emph{volume} of $S$ is defined as $vol(S) := \sum_{v \in S} \deg(v)$. 
The \emph{conductance} of $S$ is defined as $\phi(S) := e(S,V \setminus S)/\min(vol(S),vol(V \setminus S))$. 
If $\min(vol(S),vol(V \setminus S)) = 0$, we define $\phi(S) = 0$. 
The conductance of the entire graph $G$ is defined as $\phi_G := \min_{S \subseteq V} \phi(S)$. 
Throughout this paper, unless explicitly indicated otherwise, we adhere to the following definition of expander graphs, in which $0 \leq \phi \leq 1$. 
\begin{definition}
$G$ is a $\phi$-expander if $\phi_G \geq \phi$.
\end{definition}
\noindent
Another related notion of expansion that we use in our proofs is edge expansion. 
\begin{definition}
The edge expansion of $G$ is: 
\[
h_G := \min_{\emptyset \neq S \subseteq V, |S|\leq n/2} \frac{e(S,V \setminus S)}{|S|}.
\]
We will say that $G$ is an $h$-edge expander if $h_G \geq h$.
\end{definition}
\noindent
A known fact states that conductance and edge expansion are interchangeable in regular graphs.
\begin{fact}
\label{prop:expansion}
If $G$ is a $d$-regular $h$-edge expander, then it is also a $\frac{h}{d}$-expander. 
In particular, if $G$ is an $\Omega(d)$-edge expander, then it is also an $\Omega(1)$-expander.
\end{fact}

\section{The Core Gadget}
\label{sec:coregadget}
In this section, we present a derandomized core gadget on which we base our results. 
The core gadget is a deterministic algorithm that takes a graph $G$ and augments it with new vertices and edges to output an 
$\Omega(1)$-expander $G_{exp}$ in $\tilde{O}(m+n)$ time. 
We will utilize an explicit construction of $d$-regular, bipartite, $\phi d$-edge expanders on $2N$ vertices for some constant $\phi > 0$. Assume that given any $d \geq 3$, and sufficiently large $N$, we can construct such graph in $\tilde{O}(Nd)$ time. See \cref{app:explicit} for details about the construction. 

\subsection{The core gadget}
Given a graph $G$, augment $G$ with an initially empty bipartite graph 
featuring $N$ vertices on each side, denoted $L := \{x_1, x_2, \ldots, x_N \}$ and $R := \{ y_1, y_2, \ldots, y_N \}$. 
$N$ is a parameter that can vary depending on the application, but for the most basic construction, we set it to $N = (1+o(1))n$. 
The gadget consists of the following two-step construction:
\begin{enumerate}
    \item (Degree balancing) For every $v \in V$ add, using a Round-Robin algorithm, $\deg_G(v) + 3$ neighbors in $L$. Namely, follow a circular order over $L$ to pick neighbors one at a time. Assuming $N \geq n+2$, every vertex is connected to $\deg_G(v)+3$ \emph{distinct} neighbors in $L$ without wrapping around. 
    The Round-Robin algorithm guarantees that the degrees within $L$ are almost balanced. 
    In more detail, since the total number of edges that we add in this step is $\sum_{v \in V} (\deg_G(v) + 3) = 2m + 3n$, 
    the degree of every vertex in $L$ is either $\lfloor \frac{2m + 3n}{N} \rfloor$ or $\lceil \frac{2m+3n}{N} \rceil$. 
    \item (Expander construction) Set $d := \lceil \frac{2m+3n}{N} \rceil$. Let $X$ be a bipartite, $d$-regular, $\phi d$-edge expander on $L \cup R$, constructed using the algorithm from 
    \cref{app:explicit}. 
\end{enumerate}
\noindent
Denote the resulting graph by $G_{exp} := (V_{exp}, E_{exp})$, illustrated in \cref{fig:coregadget}. 
Note that the running time of this gadget is $\tilde{O}(m + n + dN) = \tilde{O}(m+n)$. 

\begin{figure}[h] 
    \centering 
    \includegraphics[scale=0.5]{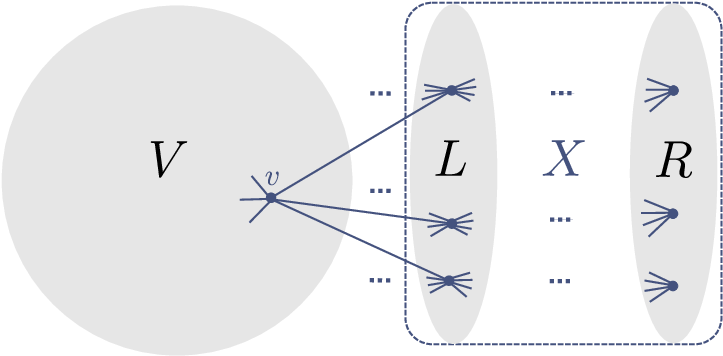}
    \caption{The core gadget augments $G$ with $O(m+n)$ edges connected to a bipartite, $d$-regular, $\phi d$-edge expander, resulting in an $\Omega(1)$-expander $G_{exp}$.}
    \label{fig:coregadget}
\end{figure}
\noindent

\begin{lemma}
    \label{lem:coregadget}
    The graph $G_{exp}$ is an $\Omega(1)$-expander. 
\end{lemma}

\noindent
The proof of this lemma, on the high level, splits into two cases. In the first case, we deal with cuts whose volume is mostly concentrated in $X$, and show that there must be many out-going edges because $X$ is an expander and the degrees in $L$ are bounded. In the second case, we deal with cuts whose volume is mostly concentrated in $G$ and show that there must be many out-going $V$-to-$L$ edges in this case. Let us now prove this claim formally. 
For clarity, we henceforth omit the subscript $G_{exp}$ from $e_{G_{exp}}(\cdot,\cdot)$, and keep the subscript $G$ in $e_{G}(\cdot,\cdot)$. Similarly, we do for $E(\cdot,\cdot)$, $\deg(\cdot)$, and $vol(\cdot)$. 

\begin{proof}[Proof of \cref{lem:coregadget}] 
\label{proof:lemma9}
    To prove that $G_{exp}$ is an $\Omega(1)$-expander, consider a non-empty cut $S \subseteq V_{exp}$. 
    We need to show that $\min(vol(S),vol(V_{exp} \setminus S)) > 0$, and 
    that $e(S,V_{exp} \setminus S) \geq \phi \cdot \min(vol(S),vol(V_{exp} \setminus S))$ for some constant $\phi>0$. 
    It is clear that $\min(vol(S),vol(V_{exp} \setminus S)) > 0$ because all vertices have a degree of at least $3$. 
    Let $S_V := S \cap V$, $S_L := S \cap L$, and $S_R := S \cap R$. 
    Assume w.l.o.g.\ that $|S_L \cup S_R| \leq N$, as it holds either for $S$ or its complement. 
    Note that we have $vol(S) = vol(S_V) + vol(S_L \cup S_R)$, and consider the following cases:
    \begin{itemize}
        \item $vol(S_V) \leq 4vol(S_L \cup S_R)$. In this case, $vol(S) \leq 5vol(S_L \cup S_R)$, and since every $x \in L$ has degree $2d$ or $2d-1$, and every $y \in R$ has degree $d$, then 
        $5vol(S_L \cup S_R) \leq 10d |S_L \cup S_R|$. Since $X$ is a $d$-regular, $\phi d$-edge expander on $2N$ vertices, 
        and $|S_L \cup S_R| \leq N$, there is some constant $\phi>0$ such that $e(S_L \cup S_R, (L \cup R) \setminus (S_L \cup S_R)) \geq \phi d |S_L \cup S_R|$. Hence, 
        $e(S,V_{exp} \setminus S) \geq \phi d |S_L \cup S_R| \geq \phi vol(S)/10$, and we are done.
        \item $vol(S_V) > 4vol(S_L \cup S_R)$. In this case, $vol(S) < 5 vol(S_V)/4$. We show that there are many $S_V$-to-$(L \setminus S_L)$ edges. Note that $e(S_V,L) \geq vol(S_V) / 2$, because for 
        every vertex $v \in V$, there are more $v$-to-$L$ edges than $v$-to-$V$ edges. Hence:
        \begin{equation}
        \label{eq:conductance2}
        \begin{split}
        &e(S_V, L \setminus S_L) = e(S_V,L) - e(S_V,S_L) \geq 
         \frac{vol(S_V)}{2} - vol(S_L \cup S_R) > \\ 
        &\frac{vol(S_V)}{2} - \frac{vol(S_V)}{4} = 
        \frac{vol(S_V)}{4} \geq vol(S)/5, 
        \end{split}
        \end{equation}
    and we are done.
    \end{itemize}
\end{proof}
\begin{figure}[h]
    \centering 
    \includegraphics[scale=0.4]{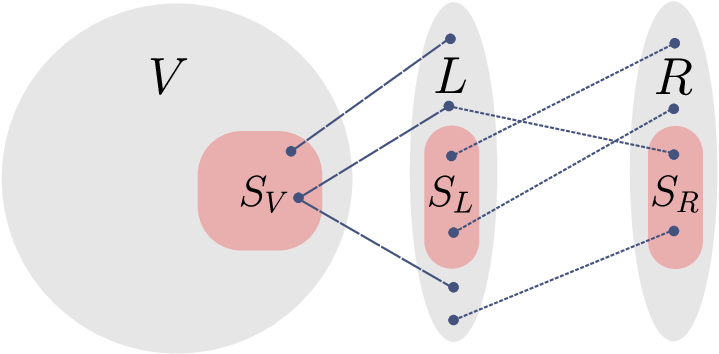}
    \caption{
    An illustration for the proof of \cref{lem:coregadget}. In the first case, we deal with cuts whose volume is mostly concentrated in $X$ and show that there are many short-dashed edges because $X$ is an expander, and the degrees in $L$ are bounded. In the second case, we deal with cuts whose volume is mostly concentrated in $G$ and show that, in this case, there must be many long-dashed edges.}
    \label{fig:gadgetproof}
\end{figure}

We remark that the blowup of the core gadget, i.e., the number of added vertices and edges, is $2N = O(n)$ vertices and $O(Nd) = O(m+n)$ edges. 
Furthermore, if $G$ is a graph of maximum degree $\Delta$, then $G_{exp}$ is a graph of maximum degree $2 \Delta + 3$.

\section{Variants of the Core Gadget}
\label{sec:variants}
In this section, we discuss a few variants of the core gadget and how to adapt it to the fully dynamic setting. 
\subsection{Robust core gadget}
The first variant we consider is a generalization of the core gadget, which allows greater variations in the degrees of the vertices in $G_{exp}$. 
This generalization will be useful in the analysis of various Direct-WTERs both in the static and in the dynamic setting. 
Namely, the next lemma shows that changing the degrees by a constant factor reduces the conductance by, at most, a constant factor. 
\begin{lemma}
    \label{lem:robust}
    For every $0<\eps\leq 1$, $\alpha \geq 1$, and an integer $d_X \geq 3$, consider the following generalization of the core gadget: 
    \begin{enumerate}[label=(\arabic*)]
        \item Every $v \in V$ has at least $\eps \deg_G(v)+1$ neighbors in $L$. 
        \item $X$ is a $d_X$-regular, $\phi d_X$-edge expander, for some some constant $\phi > 0$. 
        \item The degrees of all the vertices in $L$ are within $[d_X, \alpha d_X]$.
    \end{enumerate}
    Then $G_{exp}$ is an $\phi \eps /(5 \alpha)$-expander. 
\end{lemma}

\begin{proof}
    The proof is obtained by adapting the proof of \cref{lem:coregadget} to this generalization. For brevity, we only discuss its differences from the proof of \cref{lem:coregadget}. 
    Note that $vol(S) > 0$ for every $S \subseteq V_{exp}$ because all vertices have positive degrees. Let us adapt the two cases as follows:
    \begin{itemize}
        \item $vol(S_V) \leq \frac{4}{\eps} vol(S_L \cup S_R)$. In this case, $vol(S) \leq (1+\frac{4}{\eps}) vol(S_L \cup S_R) \leq (1+\frac{4}{\eps})\alpha d_X |S_L \cup S_R|$. Therefore, since $X$ is a $\phi d_X$-edge expander, we get:
        \begin{equation*}
        \begin{split}
            &e(S_L \cup S_R, (L \cup R) \setminus (S_L \cup S_R)) \geq 
            \phi d_X |S_L \cup S_R| \geq \frac{\phi}{\alpha(1+\frac{4}{\eps})} vol(S) \geq \frac{\phi \eps}{5\alpha} vol(S), 
        \end{split}
        \end{equation*}
        and we are done. 
        \item $vol(S_V) > \frac{4}{\eps} vol(S_L \cup S_R)$. In this case, $vol(S) \leq (1+\frac{\eps}{4}) vol(S_V)$. Note that $e(S_V,L) \geq  \frac{\eps}{\eps+1} vol(S_V) \geq \frac{\eps}{2} vol(S_V) $, because every $v \in V$ has $\deg_G(v)$ neighbors in $V$ and at least $\eps \deg_G(v)$ neighbors in $L$. Hence:
        \begin{equation}
        \begin{split}
            &e(S_V, L \setminus S_L) = e(S_V,L) - e(S_V,S_L) \geq 
             \frac{\eps}{2} vol(S_V) - vol(S_L \cup S_R) > \\ 
            &\frac{\eps}{2} vol(S_V) - \frac{\eps}{4} vol(S_V) = 
            \frac{\eps}{4} vol(S_V) \geq \eps vol(S)/5, 
        \end{split}
        \end{equation}
        and we are done.
    \end{itemize}
\end{proof}

\subsection{Fully-dynamic core gadget} 
\label{sub:dyn}
Our algorithm makes use of three procedures: (1) $\tt{UPDATE}$, which computes and adds a batch of edges from a vertex in $G$ to $L$, (2) $\tt{BALANCE}$, which rebalances the degrees in $L$ using Round-Robin, and (3) $\tt{RECOMPUTE}$, which recomputes the graph using the static core gadget. 
In this section, we use $G^{exp}$ to denote a dynamic $\Omega(1)$-expander output by the dynamic core gadget. 
We use subscripts to indicate the state of a dynamic graph at a certain time, e.g., $G_t$ is the dynamic graph $G$ at time $t$. 
Let us describe the algorithm. 
\noindent

\paragraph{Preprocessing}
In the preprocessing step, given $G_0$, apply the static core gadget to construct an $\Omega(1)$-expander $G_0^{exp}$. 
Store the vertices of $L$ sorted according to their degrees, in a data structure supporting updates and successor queries in $O(1)$ time.\footnote{Na\"ively this would take $\tilde{O}(1)$ time using standard data structures, 
but it can be optimized to $O(1)$ since the degrees are integers in the range $[1,2N]$, and our updates only increase or decrease the degree of a vertex by $1$.}

\paragraph{Edge insertions and deletions}
For every insertion of an edge $uv$ to $G$, begin by inserting $uv$ to $G^{exp}$.
Denote by $m_t$ the number of edges in $G$ at the last time we applied the $\tt{RECOMPUTE}$ procedure (or $m_0$ if we did not apply it yet). 
If $m \geq 2m_t + n$, apply the $\tt{RECOMPUTE}$ procedure to recompute the graph using the static core gadget and finish.

Otherwise, let us describe the process we apply to $v$ and similarly do to $u$. 
Denote by $\deg_L(v)$ the number of neighbors that $v$ has in $L$, i.e., $\deg_L(v) = |N(v) \cap L|$. 
\begin{enumerate}[label=(\alph*)]
    \item If $\deg_{G}(v) < 2 \deg_L(v)$, finish. Otherwise, apply the $\tt{UPDATE}$ procedure to $v$ to add a new batch of neighbors of $v$ in $L$, after which we have $\deg_L(v) = \deg_G(v) + 3$. 
    \item Check if the degrees in $L$ became unbalanced, namely, 
    if $\Delta_L  \geq 2 \delta_L$, where $\Delta_L$ and $\delta_L$ are the maximum and minimum-degree vertices in $L$, respectively.
    If so, apply the $\tt{BALANCE}$ procedure, after which $\Delta_L\in \{\delta_L, \delta_L + 1\}$. 
\end{enumerate}
\noindent
For every deletion of an edge $uv$ from $G$, delete $uv$ from $G^{exp}$. Then, if $n \leq m \leq 0.5 m_t$, apply the $\tt{RECOMPUTE}$ procedure to recompute $X$. 
\noindent

Let us now describe the three procedures used above. 
\begin{itemize}
    \item $\tt{UPDATE}$ Let $k := \deg_L(v)$. Compute $k+3$ minimum-degree vertices in $L \setminus N(v)$, denoted $x_1, x_2, \ldots, x_{k+3}$, 
    by repeatedly making successor queries to the minimum degree vertex in $L$ and skipping vertices which belong to $N(v)$. For every $x_i$, insert an edge $v x_i$ to $G^{exp}$. 
    Note that now we have $\deg_L(v) = \deg_G(v) + 3$. 
    \item $\tt{BALANCE}$
    Compute, using Round-Robin, a new set of $V$-to-$L$ edges, denoted $A$, and then replace $E(V,L)$ with $A$. 
    This is done by inserting the edges of $A$ in the order given by the Round-Robin algorithm before removing $E(V,L) \setminus A$, to ensure that the degrees do not vary too much 
    in the intermediate graphs. 
    \item {$\tt{RECOMPUTE}$}. 
    Apply the static core gadget to compute a set of edges $A$, and $X'$, where $A$ is the set of $V$-to-$L$ edges define above, and $X'$ is the expander on $L \cup R$. 
    In particular, $X'$ is a $d_{X'}$-regular, $\Omega(d_{X'})$-edge expander for $d_{X'} =\lceil \frac{2m + 3n}{N} \rceil$. Replace $X$ with $X'$ by first inserting the edges of $X'$, 
    and then removing the leftover edges of $X$ which do not belong to $X'$. Then insert the edges of $A$ and remove $E(V,L)$, as we did above. 
\end{itemize}

Let us now analyze the dynamic core gadget. 

\begin{proposition}
    Every intermediate graph output by the dynamic core gadget has blowup $O(m+n)$ and conductance $\Omega(1)$. 
\end{proposition}
\begin{proof}
    The blowup is clearly $O(m+n)$ because at every intermediate step, we add at most $|A| + |e(X')| = O(m+n)$ edges. 
    Let us prove that the conductance is $\Omega(1)$. Observe that:
    \begin{itemize}
        \item At any stage, we have $\deg_L(v) \geq \max(3,\deg_G(v)/2)$, which implies that $\deg_L(v) \geq \deg_G(v)/5+1$. 
        \item At any stage of $\tt{RECOMPUTE}$, the edge-expansion of the graph is at least $\phi d_{X'}/4$ because the degree $X$ (the expander we discard) is at least $\phi d_{X'}/4$. Note that adding edges to an 
        edge expander does not ruin edge expansion. 
        \item At any stage, the degrees in $L$ are within $[d_X, 4d_X]$ for some constant $\alpha$. This follows since $\delta_L \geq d_X$ and $\Delta_L \leq 4 \delta_L$. 
    \end{itemize}
    Therefore, by \cref{lem:robust}, the graph is an $\Omega(1)$-expander. 
\end{proof}

\noindent
Now, let us analyze the performance of the core gadget. 
\begin{proposition}
    The preprocessing time of the core gadget is $\tilde{O}(m+n)$, and the amortized update time is $O(1)$. In particular, for every sequence of $M$ updates to $G$, it computes in $\tilde{O}(m+n) + O(M)$ time a sequence of $O(M)$ updates to $G_{exp}$. 
\end{proposition}
\begin{proof}
    We prove this proposition for each of the above procedures separately that their amortized cost is $O(1)$:
    \begin{itemize}
        \item $\tt{UPDATE}$. The cost of this step is bounded by $O(k)$ (where $k = \deg_L(v)$) because there are at most $k + (k+3)$ successor queries to the data structure, 
        and at most $k+3$ edge insertions to $G^{exp}$. Note that $k$ might not be a constant. However, since $\deg_G(v) = 2k$ and $\deg_L(v) = k$, and after every call to $\tt{UPDATE(v)}$, we have $\deg_L(v) = \deg_G(v) + 3$, there are $k+3$ edge insertions of edges incident to $v$ that we can charge to:
        the edges incident to $v$ that were inserted to $G$ since the last call to $\tt{UPDATE(v)}$.
         We charge the cost of {\tt UPDATE} to them. Since there are at least $k$ many edges, and the cost of {\tt UPDATE} is $O(k)$, the amortized cost is $O(1)$. 
        \item {\tt BALANCE}. The cost of {\tt BALANCE} is $O(m+n)$. To bound the amortized cost, we will show that every call to {\tt BALANCE} is preceded by 
        $\Omega(n+m)$ edge insertions to $G$. 
        To see why, consider a graph obtained from a call to {\tt RECOMPUTE} (or preprocessing), in which for every vertex $v \in V$, we have $\deg_L(v) = \deg_G(v) + 3$. 
        We claim that during edge insertions that follow, every insertion causing an increase in $\Delta_L$ (through {\tt UPDATE}), except the first one, is preceded by at least $N/4$ edge insertions to $G$. 
        To see why, let us denote by $N_i$ the number of vertices of degree $i$ in $L$. 
        Consider a call to $\tt{UPDATE(v)}$ that increases $\Delta_L$ or $N_{\Delta_L}$ (or both). As we observed before, there are at least $k+3 = \deg_L(v)+3$ edges we can charge to. 
        We will either charge to those edges or reserve them to be charged later. Namely:
        \begin{itemize} 
            \item For every increase in $N_{\Delta_L}$: we reserve one uncharged edge to be charged later. Since {\tt{UPDATE}} can increase $N_{\Delta_L}$ by at most $k+3$, 
            the number of uncharged edges is sufficiently large. Hence, at any point in time, there is a reserve of $N_{\Delta_L}$ uncharged edges in the system.
            \item If {\tt UPDATE} increases $\Delta_L$: it means that $N-k-N_{\Delta_L} < k+3$, because the number of neighbors added to $v$, $k+3$, is greater than the number of vertices in 
            $L \setminus N(v)$ that have degree $< \Delta_L$, and there are at least $N-k-N_{\Delta_L}$ such vertices. Now, if $N_{\Delta_L} < N/2$, then $N/4 \leq  k$, so we charge the increase in $\Delta_L$ 
            to $k \geq N/4$ edge insertions. On the other hand, if $N_{\Delta_L} \geq N/2$, then we have a reserve of $\geq N/2$ uncharged edges. We charge the increase in $\Delta_{L}$ to this reserve. 
        \end{itemize}
        \noindent
        Therefore, since we apply {\tt BALANCE} only when $\Delta_L \geq 2 \delta_L$, starting with $\Delta_L \in \{ \delta_L, \delta_L+1 \}$, 
        there are at least $\delta_L N/ 4 = \Omega(m_t + n) = \Omega(m+n)$ edge insertions before that. 

        \item {$\tt{RECOMPUTE}$}. The cost of {\tt RECOMPUTE} is $O(n+m)$. Consider the potential function $|m-m_t|$. After a call to {\tt RECOMPUTE} it is always $0$. 
        Now, if the call to {\tt RECOMPUTE} followed an edge insertion, then $m = 2m_t + n$. If the call to {\tt RECOMPUTE} followed an edge deletion, then $m = 0.5m_t \geq n$. 
        In both cases, the difference between the potentials after and before the call is negative and proportional to the cost. Therefore, the amortized cost is $O(1)$. 
    \end{itemize}
\end{proof}

\subsection{Blowup-conductance tradeoff}
\label{sub:tradeoff}
Abboud and Wallheimer~\cite{abboud2023worst} observed that in some cases, e.g., for exponential-time problems and approximation problems, even a linear blowup might be too costly. Hence, they suggested a variant of their core gadget that gives a tradeoff between the number of vertices and the conductance of $G_{exp}$. Their idea can be adapted to our core gadget as follows. Given $0 < \eps \leq \delta\leq 1$, modify the core gadget by setting $N = \lceil \delta n (1+o(1)) \rceil$, and for every $v \in V$, instead of adding $\deg_G(v)+3$ edges from $v$ to $L$, add $\lceil \eps \deg_G(v) \rceil+3$ such edges. By \cref{lem:robust}, the resulting graph is an $\Omega(\eps)$-expander. Moreover, the blowup in the number of vertices is $2N \leq 2\delta n + O(1)$, and the blowup in the number of edges is $2 \eps m + 3n$. Rescaling appropriately, we obtain that given any $0 < \eps \leq \delta \leq1$, we can construct an $\Omega(\eps)$-expander with a blowup of $\delta n$ in the number of vertices and $\eps m + 3n$ in the number of edges. 

\subsection{Bipartiteness-preserving core gadget} 
In some instances, we would like our core gadget to preserve bipartiteness, e.g., when the problem is defined for bipartite graphs. Note that our core gadget does not preserve bipartiteness because the endpoints of any odd-length path in $G$ may share a common neighbor in $L$, thus closing an odd cycle and making 
$G_{exp}$ non-bipartite. Therefore, we suggest a variant of the core gadget that, given a bipartite graph $G$, outputs a bipartite, $\Omega(1)$-expander $G_{bexp} := (V_{bexp},E_{bexp})$. 

Given a bipartite graph $G = (A \cup B,E)$, where $|A| = |B| = n$, let us modify the core gadget as follows. Instead of adding edges from $A \cup B$ to $L$, add edges from $A$ to $L$ 
and from $B$ to $R$. Namely, for every $v \in A$, add $\deg_G(v) + 3$ edges to $L$, and for every $v \in B$, add $\deg_G(v) + 3$ edges to $R$. The rest of the construction stays the same. Namely, 
we construct a $d$-regular expander $X$ between $L$ and $R$ for the same $d$ as before. Clearly, the blowup in the number of edges and vertices is still linear, and the graph is bipartite, with sides $A \cup R$ and $B \cup L$. In addition:

\begin{claim} 
    \label{prop:bexp} 
    $G_{bexp}$ is a bipartite $\Omega(1)$-expander. 
\end{claim} 
\noindent
To see why this claim holds, let $A \cup B$ play the role of $V$, and let $L \cup R$ play the role of $L$ in the original proof of \cref{lem:coregadget}.

\section{Direct-WTERs for Max-Cut, Densest Subgraph, and Graphical OMv}
\label{sec:wters}

In this section, we develop further this line of research by providing a Direct-WTER for Max-Cut, a DD-WTER for Densest Subgraph, and a DD-WTER for Graphical OMv instances.

\subsection{Direct-WTER for Max-Cut}
\label{sub:maxcutwter}
In the Max-Cut problem, the goal is to compute the maximum cut in a graph $G$, which we denote by $MC(G) := \max_{S \subseteq V} e(S,V \setminus S)$. 
To make the graph an $\Omega(1)$-expander, simply applying the core gadget does not work because it might affect the maximum cut in $G_{exp}$ unpredictably. 
To this end, we add a gadget that ensures that any maximum cut in $G_{exp}$ will separate $L$ from $R$. Additionally, we modify the core gadget so that every vertex in $G$ 
will have the same number of neighbors in $L$ and $R$. These two gadgets together ensure that the vertices in $G$ do not get a preference to be in the part of $L$ or $R$, due to symmetry.
Consequently, the maximum cut in $G_{exp}$ induces a maximum cut in $G$. 

Let us now describe the reduction in more detail. Given $0 <\eps \leq 1$, apply the core gadget with a tradeoff between the conductance and the blowup, as described in \cref{sub:tradeoff}, 
to get an $\Omega(\eps^2)$-expander with a blowup of $2N = \eps n$ in the number of vertices. 
This is achieved  by picking $N = \eps/2 n$ (roughly), and adding for every $v \in V$, $\eps^2 \deg_G(v) + 3$ neighbors in $L$. 
Add a symmetric copy of the $V$-to-$L$ edges between $V$ and $R$ as well. 

Observe that for every vertex in $L \cup R$, the number of neighbors it has in $V$ is at most $d \leq \frac{2\eps^2 m + 4n}{N} = 4 \eps m/n + O(1) \leq 4 \eps n + O(1)$, 
where $d$ is also the degree of the expander between $L$ and $R$. 
Now, add two bi-cliques of size $N \times 3d$ as follows: add two sets of $F_L$ and $F_R$ containing $3d$ vertices each, and all $L$-to-$F_L$ and $R$-to-$F_R$ edges. 
The purpose of this gadget is to ensure that $L$ and $R$ are separated in a maximum cut. Denote the resulting graph by $G_{exp}$. See Figure \cref{fig:max-cut}.

\begin{figure}[h]
    \centering
    \includegraphics[scale=0.5]{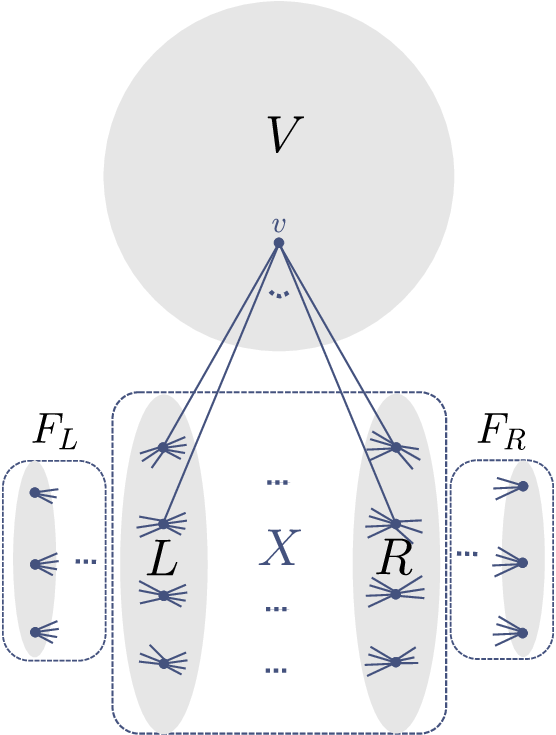}
    \caption{A Direct-WTER for the Max-Cut problem.}
    \label{fig:max-cut}
\end{figure}
\noindent

Observe that the blowup in the number of vertices is $2N + 6d \leq 25 \eps n + O(1)$. 
Moreover, we claim that the additional gadgets: the $V$-to-$R$ edges and the bi-cliques, do not ruin the graph's expansion, and that $G_{exp}$ is an $\Omega(\eps^2)$-expander.
To see why, observe that the induced graph on $L \cup F_R \cup F_L \cup R$ is a bipartite, $\Omega(d)$-edge expander, 
hence the proof of \cref{lem:robust} holds for it as well (by letting $L \cup R$ play the role of $L$).

Finallly, the next claim shows that the maximum cut in $G_{exp}$ encodes the maximum cut in $G$. 
\begin{claim}
    \label{claim:maxcutproof}
    The maximum cut in $G_{exp}$ is $MC(G_{exp}) = MC(G) + 7dN$. 
\end{claim}
\begin{proof}[Proof of \cref{claim:maxcutproof}]
    Let $(S,V_{exp} \setminus S)$ be a maximum cut in $G_{exp}$. We call the vertices of $S$ \emph{red} and the vertices of $V_{exp} \setminus S$ \emph{blue}. 
    Consider the following observations: 
    \begin{enumerate}
    \item $F_L$ must be monochromatic because otherwise, we could improve the cut by giving $F_L$ the opposite color of the majority color in $L$. A similar observation applies to $F_R$ as well. 
    \item $L$ is monochromatic and has the opposite color to $F_L$. To see why, assume w.l.o.g. that $F_L$ is red, and for contradiction, assume that some $x \in L$ is also red. 
    By changing the color of $x$ to blue, we gain $|F_L| = 3d$ edges and lose at most $\deg_V(x) + \deg_R(x) \leq 2d$ edges, a contradiction to $S$ being a maximum cut.
    A similar observation applies to $F_R$ as well.
    \item $L$ and $R$ have different colors. To see why, assume w.l.o.g. that $L$ is red, and assume for contradiction that $R$ is also red. 
    Based on the previous observation, $F_L$ is blue. 
    By flipping the colors of $L$ and $F_L$, we gain the $L$-to-$R$ edges; there are $N \cdot d$ such edges, and lose the $L$-to-$V$ edges; there are at most $N \cdot d$ such edges. Hence, we can assume without loss of generality that this observation holds. 
    \end{enumerate}
    Now, we claim that $S \cap V$ is a maximum cut in $G$. To see why, observe that since $L$ and $R$ have different colors, and for every $v\in V$, we have $\deg_L(v) = \deg_R(v)$, 
    then replacing $S \cap V$ with any other cut in $V$ will only change the number of edges cut inside $G$. Therefore, we get $MC(G_{exp}) = MC(G) + N(|F_L| + |F_R| + d) = MC(G) + 7dN$. 
\end{proof}
\noindent

\subsection{DD-WTER for Densest Subgraph}
In the Densest Subgraph problem, we define the \emph{density} of non-empty set $S \subseteq V$ to be $\rho(S) := m_S/|S|$, where $m_S$ is the number of edges in the subgraph induced by $S$. The goal is to compute $\rho(G) := \max_{S \subseteq V} \rho(S)$. 
In our DD-WTER, we will make use of the following claim: 
\begin{claim} 
    \label{claim:densest} 
    Any set $S^*$ that maximizes $\rho$ in a graph with $m$ edges and $n$ vertices does not contain any vertex of degree less than $m/n$. 
\end{claim} 
\begin{proof}[Proof of \cref{claim:densest}]
    Assume for contradiction that $S^*$ contains a vertex $u$ of degree $\deg(u) < m/n$. Define $S = S^* \setminus \{u\}$. Then, since $m/n = \rho(V) \leq \rho(S^*) = m_{S^*}/|S^*|$, we get: 
    \begin{equation*}
    \rho(S) \geq \frac{m_{S^*}-\deg(u)}{|S^*|-1} > \frac{m_{S^*}-\frac{m}{n}}{|S^*|-1} \geq \frac{m_{S^*} - \frac{m_{S^*}}{|S^*|}}{|S^*|-1} = \rho(S^*),
    \end{equation*}
    which contradicts $S^*$ being a maximizer.
\end{proof}

\noindent
In addition, we assume that the density in $G$ is sufficiently large, namely, $m > 42n$.\footnote{For clarity, we do not attempt to optimize this constant, although it can be reduced.}
To make the graph an $\Omega(1)$-expander, we apply the core gadget in a way that only introduces vertices of degree smaller than $m_{exp}/n_{exp}$, where $m_{exp}$ and $n_{exp}$ are the number of edges and vertices in $G_{exp}$, respectively. 
To this end, consider the graph obtained by applying the core gadget with a tradeoff between the conductance and the blowup, as described in \cref{sub:tradeoff}. 
Specifically, we pick the parameters such that in $G_{exp}$, every $v \in V$ is connected to $\lceil \eps \deg_G(v) \rceil + 3$ vertices in $L$, 
and $N = |L| = |R| = n(1+o(1))$. The conductance of this graph is $\Omega(\eps)$. The parameter $\eps$ is chosen to be a sufficiently small constant which will be determined later. 
\begin{claim} 
    \label{claim:densest2}
    The maximum density in $G_{exp}$ is $\rho(G_{exp}) = \rho(G)$. 
\end{claim} 
\noindent

\begin{proof}[Proof of \cref{claim:densest2}]
    We will show that $m_{exp}/n_{exp} > \mu$, where $\mu$ is the maximum degree in $L$ (and therefore also in $L \cup R$).
    Hence, by \cref{claim:densest}, it will follow that any maximum density subgraph in $G_{exp}$ does not contain vertices from $L \cup R$, so 
    it must induce a maximum density subgraph in $G$. 

    Note that the number of $V$-to-$L$ edges is bounded by $2 \eps m + 3n\leq e(V,L) \leq 2\eps m + 4n$, and that the degree of the expander between $L$ and $R$ is $d =\lceil e(V,L)/N \rceil$. Hence, we have 
    $m_{exp} = m + e(V,L) + Nd \geq m + (2\eps m + 3n) + (2\eps m + 3n) = (1+4\eps)m + 6n$. 
    In addition, note that the maximum degree in $L$ is: 
    \[
    \mu \leq 2d \leq 2\left(\frac{e(V,L) }{N}+1\right) \leq \frac{4 \eps m + 9n + o(n)}{n}. 
    \]
    Now, observe that the inequality:
    \[
    \mu \leq
    \frac{4 \eps m + 9n + o(n)}{n}  <
    \frac{(1+4\eps)m + 6n}{3n + o(n)} \leq 
    \frac{m_{exp}}{n_{exp}}, 
    \]
    holds when $n$ is sufficiently large, $1-8\eps > 0$, 
    and $21n/(1-8\eps)< m$. For example, by picking $\eps < 1/16$, we get that 
    $\mu < m_{exp}/n_{exp}$ when $m > 42n$. 
\end{proof}

To adapt the Direct-WTER to a DD-WTER, we replace the core gadget with the dynamic core gadget. However, note that in the dynamic core gadget there are some variations in the degrees due to lazy updates and rebalances. Nonetheless, these variations can be compensated for by picking a smaller $\eps$, specifically $\eps = 1/44$ will suffice.

\subsection{DD-WTER for Graphical OMv instances}
\label{sub:omvwter}
In this subsection, we demonstrate how our core gadget from \cref{sec:coregadget} 
can be used to prove OMv-hardness to various OMv-hard problems by making the typical ``OMv-hard'' instances of a problem $\Omega(1)$-expanders. 
The definitions of the OMv and OuMv problems and the OMv Conjecture appear in \cref{app:problems}. 

A \emph{Graphical OMv} instance is constructed from an OuMv instance as follows. Given a $k \times k$ binary matrix $M$, 
construct bipartite graph $G_M := (A \cup B,E)$, where $A$ and $B$ are equally-sized parts, denoted by $A = \{a_1,\ldots,a_k\}$ and $B = \{b_1,\ldots,b_k\}$. The edges of $G_M$ are defined according to the $1$'s of the matrix, i.e., $E := \{a_ib_j | M[i,j]=1 \}$. 
Next, add some problem-specific gadgets to $G_M$: for many OMv-hard problems, such as $st$-SubConn, $st$-SP, and more, the gadgets consist of 
$O(k)$ vertices and $O(k)$ edges that are connected to $L \cup B$ in a certain (dynamic) way. 

\paragraph*{Making Graphical OMv instances $\Omega(1)$-expanders.} To make such instances $\Omega(1)$-expanders, we apply the bipartiteness-preserving core-gadget on $G_M$ before adding the problem-specific gadgets. 
In some cases, as we will soon demonstrate for the $st$-SP problem, such gadgets preserve expansion, or they can be easily adapted to preserve expansion. 
Therefore, this adaptation proves that their OMv-hard instances are $\Omega(1)$-expanders. 
Let us now demonstrate this technique for the problem of $st$-SP and prove \cref{prop:st-SP} which states that $st$-SP is OMv-hard on $\Omega(1)$-expanders. 

\begin{proof}[Proof of \cref{prop:st-SP}] 
    We will prove the proposition for an easier variant of the problem called $st$-SP ($3$ vs. $5$), where the goal is only to distinguish between $dist(s,t) = 3$ and $dist(s,t) \geq 5$. 
    Henzinger et al.~\cite{henzinger2015unifying} proved a lower bound to this problem via a reduction from OuMv: construct $G_M$, add vertices 
    $s$ and $t$ to $G_M$, and then update the edges between $s$-to-$A$ and $t$-to-$B$ according to the input vectors. It then followed that whenever 
    $u^TMv = 1$, then $dist(s,t)=3$, and whenever $u^T M v = 0$, $dist(s,t) \geq 5$. By picking $k = \sqrt{m}$, the graph has $O(m)$ edges, and the lower bound the follows was $m^{1/2-\eps}$ per update 
    and $m^{1-\eps}$ per query. We now modify their construction as follows. 

    Given $M$, apply the bipartiteness-preserving core gadget to $G_M$ before adding vertices $s$ and $t$. 
    Then, pick a non-edge in the expander $X$, i.e., $xy \notin E(X)$ for some $x \in L$ and $y \in R$ arbitrarily. 
    The purpose of this modification is to ensure that $s$ and $t$ are connected to the graph without introducing a path of length $<5$ between 
    them. 
    Now, we procceed as in~\cite{henzinger2015unifying}; namely, given vectors $u = (u_1,u_2,\ldots,u_k)$ and $v = (v_1,v_2,\ldots,v_k)$, 
    we update the graph by adding the edges $sa_i$ iff $u_i =1$, and $t b_j$ iff $v_j =1$. 
    Note that $u^T M v = 1$ iff $dist(s,t) = 3$, and otherwise $dist(s,t) \geq 5$. In addition, we claim that the graph is an $\Omega(1)$-expander. 
    This follows from the next claim. 
    \begin{claim}
        \label{claim:omv}
        If $G$ is a $\phi$-expander for some $\phi >0$, then adding a vertex to $G$ and connecting it arbitrarily to $\ell \geq 1$ 
        vertices, results in an $\phi/4$-expander. 
    \end{claim}
    \begin{proof}[Proof of \cref{claim:omv}]
        Denote the resulting graph by $H$ and the added vertex by $v$. 
        Consider a non-empty cut $S \subseteq V(H)$. Denote by $S_V = S \cap V$, and assume without loss of generality that $vol_G(S_V) \leq vol_G(V \setminus S_V)$ (we can assume so because it holds either for $S$ or 
        its complement). 
        This assumption is useful since we then have $e_G(S_V, V \setminus S_V) \geq \phi vol_G(S_V)$, as $G$ is a $\phi$-expander. 
        Now, we split into the following cases. If $S$ does not contain $v$, i.e. $S_V = S$, then $vol(S) \leq vol_G(S) + |S| \leq 2vol_G(S)$, therefore 
        $e(S,V(H) \setminus S) \geq e_G(S,V \setminus S) \geq \phi/2 \cdot vol(S)$. 
        If $S$ contains $v$, then $vol(S) \leq \ell + 2vol_G(S_V)$. Now, if $\ell \leq 2|S_V|$, then $vol(S) \leq 4vol_G(S_V)$, and as in the previous case we have 
        $e(S,V(H) \setminus S) \geq e_G(S_V,V \setminus S_V) \geq \phi/4 \cdot vol(S)$. 
        If $\ell > 2|S_V|$, then $e(\{v\},V \setminus S_V) \geq \ell/2$. Therefore, we get 
        \[
        e(S,V(H) \setminus S) \geq \ell/2 + \phi vol_G(S_V) \geq \phi/2 (\ell + vol_G(S_V)) \geq \phi/4 \cdot vol(S)
        \]
    \end{proof}

    In our setting, we add two vertices to $G_M$, and at all times, each of them is connected to at least one vertex in $G_M$. Therefore, since $G_M$ is a $\phi$-expander for some constant $\phi$, 
    then the resulting graph is a $\phi/16$-expander at all times according to this claim. 
    Hence, assuming the OMv Conjecture, there is no algorithm for $st$-SP ($3$ vs. $5$) on $\Omega(1)$-expanders, 
    whose preprocessing time is polynomial, update time $m^{1/2-\eps}$, and query time $m^{1-\eps}$, for any $\eps >0$. 
\end{proof}

\section{Derandomized Direct-WTERs for Additional Problems}
\label{sec:oldwters}
In this section, we show how the deterministic core gadget and its dynamic adaptation can be combined with additional gadgets from~\cite{abboud2023worst} to obtain deterministic Direct-WTERs and  
DD-WTERs to all problems that were considered in~\cite{abboud2023worst}, and for some other related problems. 
\begin{proposition}
The following problems admit deterministic Direct-WTERs: 
Maximum Matching, Minimum Vertex Cover, $k$-Clique Detection, $k$-Clique Counting, 
Max-Clique, Minimum Dominating Set, and $H$-Subgraph Detection ($m = O(m)$ and $H$ does not contain pendant vertices). 
\\
The following problems admit DD-WTERs: Maximum Matching, Bipartite Perfect Matching, $k$-Clique Detection, $k$-Clique Counting, and $H$-Subgraph Detection ($m = O(m)$ and $H$ does not contain pendant vertices). 
\end{proposition}
\noindent
We remark that for the $k$-Clique Detection and $k$-Clique Counting problems, our additional gadget differs from the one in~\cite{abboud2023worst}. Hence, we explain them in more detail. 

\paragraph{Maximum Matching, Bipartite Perfect Matching, Minimum Vertex Cover.}
The Direct-WTER for Maximum Matching and Minimum Vertex Cover are as follows. Given $G$, apply the core gadget to obtain $G_{exp}$. Then, for every $w \in L \cup R$, add a pendant vertex connected 
to $w$ (that is, a degree-$1$ vertex $w'$ whose only neighbor is $w$). Denote by $MM(\cdot)$ the cardinality of the maximum matching in a graph. 
By~\cite[Fact 16]{abboud2023worst}, the cardinality of the maximum matching in the resulting graph is $MM(G) + 2N$. A similar claim holds for minimum vertex cover. 
By~\cite[Lemma 17]{abboud2023worst}, the graph is an $\Omega(1)$-expander. 
Notice for the Minimum Vertex Cover problem, which is an exponential time problem, we will need to use the blowup-conductance tradeoff, resulting in an $\Omega(\eps)$-expander with $n(1+\eps + o(1))$ vertices.

For the problem of Bipartite Perfect Matching, both the input and output of the reduction are bipartite graphs. 
Hence, we employ the bipartiteness-preserving core gadget on $G$ to obtain $G_{exp}$ and then add pendant vertices as before. 
Note that the resulting graph is bipartite and contains a perfect matching if and only if $G$ contains a perfect matching. 

The adaptation of this reduction to the dynamic setting for Maximum Matching is straightforward, as we need to replace the core gadget with the dynamic core gadget. For Bipartite Perfect Matching, we combine the bipartiteness-preserving variant of the core gadget to the dynamic setting by replacing the role of $L$ with $L \cup R$. Hence, Bipartite Perfect Matching admits a DD-WTER. 

\paragraph{$k$-Clique Detection, $k$-Clique Counting}
In the following reduction, we assume w.l.o.g. that the graph contains no isolated vertices, as isolated vertices can be easily dealt with. 
Begin by taking a copy of $V$, denoted $V_{ind} = \{ v_{ind} \mid v \in V \}$, in which all vertices are independent. In addition, add the following set of edges $\{ uv_{ind} | uv \in E\}$. 
Next, apply the core gadget on the graph, but apply it only to $V_{ind}$ instead of $V \cup V_{ind}$. 
Apply the core gadget on a graph whose vertex set is $V_{ind}$ $V$ and add edges only from $V_{ind}$ to $L$.
Denote the resulting graph by $\bar{G} := (\bar{V},\bar{E})$. See \cref{fig:clique}. 
Observe that for every $k \geq 3$, all $k$-cliques in $G$ appear in $\bar{G}$. In addition, any $k$-clique in $\bar{G}$ that is not contained in $G$ must contain $k-1$ vertices in $V$, and a single vertex $w_{ind} \in V_{ind}$. 
In this case, there is a corresponding $k$-clique in $G$, obtained by replacing $w_{ind}$ with $w$. From this observation, we get the following corollary:
\begin{corollary} 
Let $c_k$ denote the number of $k$-cliques in $G$. Then the number of $k$-cliques in $\bar{G}$ is $(k+1)c_k$. 
\end{corollary}
\noindent 
Therefore, this reduction is applicable for $k$-Clique Detection and $k$-Clique Counting. 

Let us now prove that this additional gadget of constructing $V_{ind}$ before applying the core gadget does not ruin the expansion by much. 
\begin{claim} 
\label{claim:vind}
$\bar{G}$ is an $\Omega(1)$-expander. 
\end{claim}
\noindent

The proof of this claim follows along the lines of the proof of \cref{lem:coregadget}. In particular, we split into two cases as follows. If a large volume is concentrated in $V$, then there 
must be many outer edges from $V$ to $V_{ind}$. We then employ \cref{lem:coregadget} to argue that there are many outer edges overall. If there is not much volume in $V$, the \cref{lem:coregadget} suffice to prove the claim. 
\begin{proof}
Consider any cut $S = S_V \cup S_{ind} \cup S_L \cup S_R$, where $S_V \subseteq V, S_{ind} \subseteq V_{ind}, S_L \subseteq L, S_R \subseteq R$, such that $|S_L \cup S_R| \leq N$. 
Since there are no isolated vertices, then $vol(S) > 0$. Observe that by \cref{lem:coregadget}, we have 
\[
e(S_V,\bar{V} \setminus S_V) \geq e(S_{ind}, L \setminus S_L) + e(S_L,R \setminus S_R) \geq \phi vol(S_{ind} \cup S_L \cup S_R), 
\]
for some constant $\phi > 0$. We now split into two cases as follows.
\begin{itemize}
\item $vol(S_V) \leq 4vol(S_{ind})$. In this case, we have $vol(S) \leq 5 vol(S_{ind} \cup S_L \cup S_R)$. 
Hence, by the above observation, we have $e(S,\bar{V} \setminus S) \geq \phi vol(S) /5$, and we are done.
\item $vol(S_V) > 4vol(S_{ind})$. In this case, we have: 
\begin{equation*}
e(S_V,V_{ind} \setminus S_{ind}) = e(S_V,V_{ind}) - e(S_V,S_{ind}) > vol(S_V)/2 - vol(S_V)/4 = vol(S_V)/4.
\end{equation*}
Hence, we get 
$e(S,\bar{V} \setminus S) \geq vol(S_V)/4 + \phi vol(S_{ind} \cup S_L \cup S_R) \geq \phi/4 vol(S)$, and we are done. 
\end{itemize}
\end{proof}

Adapting this reduction to the dynamic setting is by replacing the core gadget with the dynamic one, and making two edge updates to $uv_{ind},u_{ind}v$ for every edge update $uv$ in $G$. 
We remark that to deal with isolated vertices in the dynamic setting, we can ignore deletions that remove the last edge incident to a vertex. 
Such edges can not participate in a $k$-Clique, hence the correctness is preserved. 

\begin{figure}[h]
\centering
\includegraphics[scale=0.5]{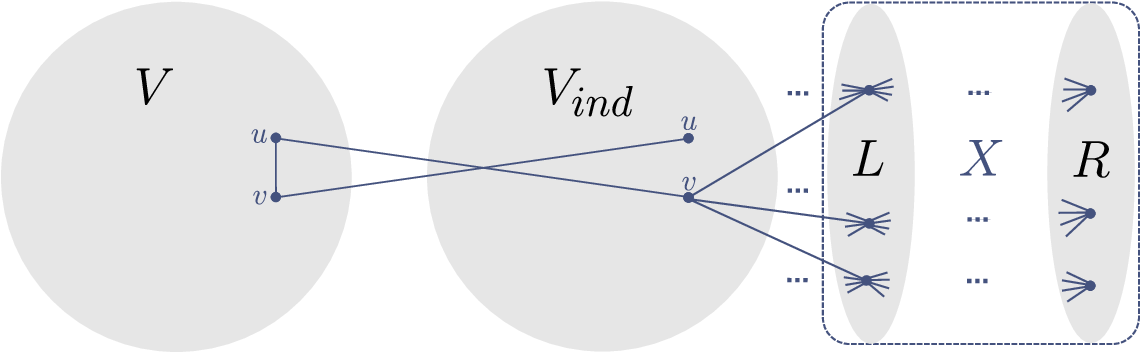}
\caption{Our Direct-WTER for $k$-Clique Detection, $k$-Clique Counting, 
and Max Clique. The dashed edges correspond to edges in $G$.}
\label{fig:clique}
\end{figure}

\paragraph{$H$-Subgraph Detection ($m = O(n)$ and $H$ does not contain pendant vertices)}
Apply the core gadget to obtain a graph $G_{exp}$. Then, subdivide each edge in $E(V,L)$ and $E(L,R)$ into a path of length $|V(H)|/2$. 
By~\cite[Claim 21]{abboud2023worst}, any copy of $H$ in the resulting graph must be a copy of $H$ in $G$. 
By~\cite[Claim 22]{abboud2023worst}, if $|V(H)| = O(1)$ then the resulting graph is an $\Omega(1)$-expander. 
As observed in~\cite{abboud2023worst}, the blowup in the number of vertices is $\Omega(m)$. Hence, it is only applicable to graphs with $m = O(n)$ edges. 
Adapting this reduction to the dynamic setting is done by replacing each edge insertion or deletion in $E(V,L) \cup E(L,R)$ by insertion of a path, resulting in a slowdown of $O(|V(H)|) = O(1)$. 

\paragraph{Max-Clique and Minimum Dominating Set}
The gadget employed for these problems is similar to the one we used for $k$-Clique Detection. Namely, we will construct the set $V_{ind}$ and then apply the core gadget. 
However, since we now deal with exponential problems, there is a need to reduce the size of $V_{ind}$ to $\eps n$. 
In~\cite[Section 4.6]{abboud2023worst}, the authors show a randomized algorithm for constructing a subset of vertices $Q \subseteq V$ of size $\eps n$, 
that ``hits'' all sparse cuts in $G$. Namely, for any cut $S \subseteq V$, either $vol_G(S \cap Q) \geq f(\eps) \cdot vol(S)$, 
or $e_G(S,V \setminus S) > f(\eps) vol(S)$, where $f(\eps) = \frac{\eps^2}{10\log(1/\eps)}$.
The only randomized component of this construction is to pick a random subset $Q \subseteq V$ of cardinality $\eps n$, such that for every vertex $v \in V$ 
for which $\deg_G(v) > \frac{\log(1/\eps)}{\eps}$, $Q$ hits an $\eps$-fraction of the neighbors of $v$.This component can be derandomized, e.g., by exhaustive search, 
incurring a (affordable) slowdown of $O(2^{\eps n})$ time. 

For the problem of Max-Clique, the reduction is by constructing $Q_{ind} = \{ u_{ind} | u \in Q\}$ and $E(V,Q_{ind}) = \{ vu_{ind} | u \in Q \}$, 
and then applying the generalized core gadget with respect to $Q_{ind}$ to obtain a conductance-blowup tradeoff. Namelt, we connect every $w_{ind} \in Q_{ind}$ to $\lceil \eps \deg(w_{ind})\rceil+3$ neighbors in $L$, and set $N = \lceil\eps n \rceil+o(n)$. Denote the resulting graph by $\bar{G}$. 
It follows from our previous observations that any $k$-clique in $\bar{G}$ corresponds to a $k$-clique in $G$. Hence, 
the maximum clique in $\bar{G}$ has the same size as the maximum clique in $G$. The blowup in the number of vertices in $\bar{G}$ is $2 \eps n + o(n)$. In addition:
\begin{claim} 
$\bar{G}$ is an $\Omega(f(\eps))$-expander. 
\end{claim} 
The proof splits into the following cases. Either the cut $S_V$, is hit by $Q$, in which case there are many edges to $Q_{ind}$, and the proof continues by combining the proofs of 
\cref{claim:vind} and \cref{lem:robust}. 
Otherwise, by the properties of $Q$, there must be many outer edges inside $V$ that account for $vol(S_V)$. Then, the proof follows from the previous results. 
\begin{proof} 
Consider any cut $S = S_V \cup S_{Q} \cup S_L \cup S_R$, where $S_V \subseteq V, S_{Q} \subseteq Q_{ind}, S_L \subseteq L, S_R \subseteq R$, such that $|S_L \cup S_R| \leq N$. 
Clearly, $vol(S) > 0$, assuming there are no isolated vertices in $G$. Observe that by \cref{lem:robust}, we have that 
\[
e(S_Q \cup S_L \cup S_R) \geq \eps \phi vol(S_Q \cup S_L \cup S_R) 
\]
for some constant $\phi > 0$. 
We split into the following cases: 
\begin{itemize} 
\item $vol_G(S_V \cap Q) < f(\eps) vol_G(S_V)$. In this case, by the properties of $Q$ we have that $e(S_V, V \setminus S_V) \geq f(\eps) vol_G(S_V) \geq f(\eps) vol(S_V)/2$. 
Hence, 
\begin{equation*}
\begin{split}
&e(S,\bar{V} \setminus S) \geq f(\eps) vol(S_V)/2 + e(S_Q,L \setminus S_L) + e(S_L,R \setminus S_R) \geq \\
&f(\eps) vol(S_V)/2 + \phi vol(S_Q \cup S_L \cup S_R) \geq \phi f(\eps) vol(S) /2. 
\end{split}
\end{equation*}

\item $vol_G(S_V \cap Q) \geq f(\eps) vol_G(S_V)$. In this case, by the construction we have that $e(S_V,Q) \geq f(\eps) vol_G(S_V) \geq f(\eps) vol(S_V)/2$. 
We split into two more cases, as follows. If $vol(S_V) < 4 vol(S_Q)/f(\eps)$, then by the above observation we have:
\[ 
e(S_Q \cup S_L \cup S_R) \geq \phi vol(S_Q \cup S_L \cup S_R) > \phi f(\eps) vol(S)/4.
\] 
Otherwise, we have that $e(S_V,Q \setminus S_Q)$ must be large:
\[
e(S_V,Q \setminus S_Q) = e(S_V,Q) - e(S_V,S_Q) \geq  f(\eps) vol(S_V)/2 - f(\eps) vol(S_V)/4 \geq f(\eps) vol(S_V)/4. 
\]
As in the previous case, we get $e(S,\bar{V} \setminus S) \geq f(\eps) \phi vol(S_V)/4$.
\end{itemize}
\end{proof} 

The reduction for Minimum Dominating Set is the same, but with the additional gadget of adding pendant vertices to the vertices in $L$. 
Denote the cardinality of the Minimum Dominating Set by $MDS(\cdot)$. 
By~\cite[Fact 23]{abboud2023worst}, it follows that the cardinality of the minimum dominating set in $\bar{G}$ is $MDS(G) + |L| = MDS(G) + N$. 
By~\cite[Lemma 17]{abboud2023worst}, the graph is an $\Omega(\eps)$ expander.

\section{Distributed Core Gadget}
\label{sec:distributed}

In the distributed CONGEST model, the input graph plays the additional role of serving as the communication network over which nodes can send messages of length $O( \log n )$. 
Expander decompositions in CONGEST have been utilized to obtain major breakthroughs, e.g., for Triangle Enumeration~\cite{chang2019distributed}. 
We remark that the current techniques for Direct-WTERs cannot be used in the CONGEST model. The reason is that those techniques augment the graph with additional edges, but in CONGEST, 
the edges are also communication links that cannot be added (they could be simulated, but that would incur large congestion). 
This is not so surprising since, in this model, expanders tend to be provably easier than worst-case graphs. 
For example, constructing a spanning tree in $n^{o(1)}$ CONGEST rounds~\cite{ghaffari2017distributed}, while there is a $\Omega(\sqrt{n}/\log n)$ lower bound for general graphs~\cite{peleg2000near}.

We focus on other related distributed models that may have been less studied in this context but may also benefit from expander decompositions. 
One such model is the CONGESTED-CLIQUE model, in which the input graph is separated from the communication network, which is a clique. Namely, in every round, 
each node can send $O( \log n )$ bits to every other node. Another related model which we consider is the MPC (Massively Parallel Computation) model. 
In MPC, the input graph is distributed across $M$ machines arbitrarily, each machine has a local space of $S$ words (where a word has a size of $O(\log n)$ bits), 
and we assume that the total amount of space is $M \cdot S = \Theta(n+m)$. In every MPC round, each machine can send a message consisting of $S$ words to another machine and reads at most $S$ 
words sent by other machines. We focus on \emph{sublinear-MPC}, where $S = n^{\delta}$ for some constant $0 < \delta < 1$. 

We initiate the study of WTERs in these two models. 
Our results imply that the core gadget can be adapted to both models, resulting in Direct-WTERs in CONGESTED-CLIQUE and Sublinear-MPC for all problems that were considered in the static and dynamic settings. 
Namely: 
\begin{proposition}
The following problems admit a deterministic Direct-WTER in CONGESTED-CLIQUE and Sublinear-MPC: 
Maximum Matching, $k$-Clique Detection, $k$-Clique Counting, Densest Subgraph ($m \geq 48n$), Minimum Dominating Set, Minimum Vertex Cover, Max-Cut, 
$H$-Subgraph Detection (where $m = \tilde{O}(n)$ and $H$ does not contain pendant vertices). 
\end{proposition}

\paragraph{CONGESTED-CLIQUE}
We will refer to the basic components of the input graph as vertices and to the basic components of the communication network as nodes. 
At the beginning, every node corresponds to a vertex. The goal is to simulate the graph $G_{exp}$. 
Some nodes in the network will play two roles: the role of their original vertex (some $v \in V$) and the role of some vertex in $L \cup R$. 
In the first round, a special node is chosen to receive all the degrees of all the vertices in the graph. The special node simulates Round-Robin locally and computes indices $\{ i_v \}_{v \in V}$ corresponding to the starting positions given by the Round-Robin for every vertex in $V$ (i.e. $i_v$ is the position such that vertex $v$ is connected in $L$ to $x_{i_v},\ldots,x_{i_v+\deg_G(v)+2}$). It 
then sends them back in the next round so that every node corresponding to $v \in V$ has a list of neighbors that $v$ has in $G_{exp}$.
In addition, the special node computes the number of edges in $G$ and sends them to the nodes playing the role of $L \cup R$. Each of these nodes can compute the expander $X$ locally. 
In the final round, each node corresponding to $v \in V$ sends its ID to the nodes corresponding to its neighbors in $L$. 
The number of rounds this algorithm needs is $O(1)$, and the result of this algorithm is that every node that corresponds to $v \in V_{exp}$ knows the edges incident at $v$ in $G_{exp}$. 

\paragraph{Sublinear-MPC}
The goal is to compute the edges of $G_{exp}$ in $O(1)$ rounds so that, in the end, they are distributed across the machines. 
We will assume that the total space in the system suffices to store $G_{exp}$, which has a linear size in $G$. 
If we try to adapt the above algorithm to the MPC model, the first difficulty encountered is that we can not designate a special machine to receive all the degrees and simulate Round-Robin. 
That is because each machine is limited to receive at most $n^{\delta}$ words in every round. Hence, we execute a ``parallel Round-Robin'' as follows. 
We will assume that each vertex has a corresponding machine, which knows its degree in $G$ (multiple vertices can correspond to the same machine, which stores their degrees while respecting the space constraint). This can be computed deterministically in $O(1)$ rounds, following from~\cite{goodrich2011sorting}. 

To parallelize the Round-Robin algorithm, we simulate a communication tree of size $O(n)$, degree $n^{\delta}$, and depth $O(1/\delta)$, whose leaves correspond to $V$, 
and every node in the tree stores the sum of degrees of all the vertices in its subtree. 
To this end, the leaves of the tree are simulated using the machines that store the degrees (at most $O(M)$), 
Moreover, we will designate additional $O(n^{1-\delta})$ machines that will act as internal nodes; hence, the total number of machines is $O(M + n^{1-\delta}) = O(M)$. 
The values (degree sums) of the tree are computed in $O(1/\delta)$ rounds, starting from the leaves bottom-up. In particular, Clearly, the sum at the root is $2m$. 

In the next phase, we compute the indices $\{i_v\}_{v \in V}$ corresponding to the starting points given by the Round-Robin algorithm. This is done top-bottom: the root is given the degree-sums 
of its children and simulates Round-Robin on ``super-nodes'' corresponding to nodes in disjoint subtrees, whose degrees are the degree-sums of the subtrees. 
It then sends the starting indices to the super-nodes. Each super-node then splits its degree-sum to the sum of degrees of its children and repeats the process, starting the Round-Robin simulation 
from its starting index computed by its parent. This process terminates when the leaves, corresponding to $V$, receive the set of indices $\{i_v\}_{v \in V}$. This takes $O(1/\delta)$ rounds. The proof of correctness that the computed indices indeed correspond to the indices given by the Round-Robin algorithm is by induction on the levels of the tree. 

The system now has a set of machines that each store a subset of degrees of vertices of $V$ and their starting indices of neighbors in $L$, such that each vertex in $V$ corresponds to some machine. 
To compute the edges of $G_{exp}$ explicitly:
\begin{itemize}
\item We designate $O(M)$ machines that receive starting indices and degrees and compute the $V$-to-$L$ edges accordingly. This can be done because there are at most $O(m+n)$ such edges.
\item We designate $O(M)$ machines that compute and store the edges of $X$. This can be done since $X$ has size $O(m+n)$. 
\end{itemize}
\noindent
The resulting state of the system is that every edge of $G_{exp}$ is stored in some machine.

\section{Acknowledgements}
We would like to thank Thatchaphol Saranurak for clarifying a typo in the mentioned lower bound for the Densest Subgraph problem~\cite[Table 5]{henzinger2015unifying}. 

\bibliography{references}

\appendix

\section{OMv-hardness for Densest Subgraph in Denser Graphs} 
\label{app:densest}

In this appendix, we provide a modification to a reduction by Henzinger et al.~\cite{henzinger2015unifying}, from OuMv to the dynamic Densest Subgraph problem. 
The goal of this modification is to increase the number of edges in the output graph to be at least $C \cdot n$ for any given integer $C \geq 1$, thus showing that the Densest Subgraph $(m \geq Cn)$ 
variant of the problem is OMv-hard. Since our DD-WTER for the Densest Subgraph problem only works for sufficiently dense graphs, namely, graphs with $m > 42n$ edges, this implies that the lower bound of~\cite{henzinger2015unifying} applies even on $\Omega(1)$-expanders. 

\paragraph{The original proof of~\cite{henzinger2015unifying}.}
Let $M$ be a $N \times N$ binary matrix given as input for OuMv. 
In~\cite{henzinger2015unifying}, the authors define a \emph{bit gadget} for every $i,j$, such that if $M[i,j]=1$, then the corresponding 
bit gadget is a path on $6N$ vertices, and if $M[i,j]=0$, then it is an independent set. 
In addition, they define $N$ row gadgets and $N$ column gadgets as independent sets on $3$ vertices or triangles. One vertex of every row/column gadget is connected to the endpoints of 
$N$ corresponding bit gadgets. The row/column gadget dynamically changes between a triangle and an independent set according to the input vectors (as defined in the OuMv problem).
Denote this graph by $G$. We remark that the maximum density in $G$ is bounded by $2$.
\begin{observation}
The maximum density in $G$ is $\rho(G) \leq 2$. 
\end{observation}
\noindent
In particular, $|E(G)| \leq 2|V(G)|$. To see why this observation is true, note that every edge in this graph is incident to a degree $2$-vertex. 
Therefore, for any $S \subseteq V(G)$, we have $m_S/|S| \leq 2 x/x = 2$, where $x$ is the number of vertices of degree $2$ in $S$.

Let us now describe a self-reduction that given any instance $G$ of Densest Subgraph, such that $\rho(G) < C+1/2$ for some integer $C \geq 1$, it produces a graph $G_C$ with 
$|E(G_C)| \geq C |V(G_C)|$ edges. 
This self-reduction adapts trivially to the dynamic setting, and in particular, it can be applied to the dynamic graph described above (which has maximum density $\rho(G) \leq 2$), 
to produce OMv-hard instances with $m \geq Cn$ edges, for any constant integer $C \geq 2$, in particular $C > 42$. 

\paragraph{Our modification.}
Given some integer $C \geq 1$, attach a $(2C+1)$-clique to every vertex in $G$. Namely, for every $v \in V(G)$, construct a clique containing $v$ and additional $2C$ vertices. 
Denote the resulting graph by $G_C$. Note that the number of vertices in $G_C$ is $n = |V(G)| (2C+1)$, and the number of edges is $ m \geq |V(G)| \binom{2C+1}{2} = C n$. 
The correctness of this modification follows from the next proposition. 
\begin{proposition}
\label{prop:densest}
Let $G$ be a graph such that $\rho(G) < C+1/2$. Then the graph $G_C$, obtained from $G$ by attaching a $(2C+1)$-clique to every vertex of $G$, has maximum density 
$\rho(G_C) =\rho(G)/(2C+1) + C$. 
\end{proposition}
Hence, one can easily compute the maximum density of $G$ from the maximum density of $G_C$. Therefore, the complexity of dynamic Densest Subgraph $(m \geq C n)$ is at least the 
complexity (up to $n^{o(1)}$ factors) of dynamic Densest Subgraph when restricted to graphs with density $< C + 1/2$, e.g., the graphs output by the reduction in~\cite{henzinger2015unifying}. 
In particular, we obtain the next corollary. 
\begin{corollary}
There is no dynamic algorithm for Densest Subgraph ($m \geq C n$) whose update time is $n^{1/3-o(1)}$ and query time $n^{2/3-o(1)}$, unless the OMv Conjecture fails. 
\end{corollary}

\begin{proof}[Proof of \cref{prop:densest}]
Let $S \subseteq V(G_C)$ be a set which maximizes $\rho$ in $G_C$, i.e. $\rho(S) = \rho(G_C)$. 
We can assume w.l.o.g. that $S$ induces a connected component because all connected components induced by $S$ must have the same density. In addition, we have $\rho(S) \geq C$, 
since any $(2C+1)$-clique has density $C$. The heart of the proof is the following observation: every $(2C+1)$-clique intersected by $S$ must be fully contained in $S$. 
Essentially, we begin by proving a general claim that holds for all graphs: If some set of vertices intersects a clique but does not contain it, 
then, we can increase its density by either adding or removing a single vertex from the clique to the set, 
provided that the intersection occurs at two vertices or more.
This implies that any $(2C+1)$-clique intersected by $S$ has either one or $2C+1$ vertices in $S$, as otherwise, $S$ would not be a maximizer of $\rho$. 
Then we use the fact that $\rho(G) < C + 1/2$ to show that the former is not possible, i.e., any $(2C+1)$-clique intersected by $S$ is contained in $S$.
This shows that $S$ has the form of a densest subgraph in $G$ and attached $(2C+1)$-cliques to every vertex in contains in $G$, yielding the claimed density.

Formally, let us denote by $1 \leq k \leq 2C+1$ the size of the intersection between $S$ and a $(2C+1)$-clique which is attached to some vertex $v \in V(G)$. 
We want to show first that $k=1$ or $k=2C+1$. Assume for the sake of contradiction that $1 < k< 2C+1$. 
Note that $S$ must contain $v$, as otherwise it would not induce a connected component and have density $\geq C$. 

Since $1 < k$, there must exist some vertex $u \neq v$ in the intersection of $S$ and the attached clique of $v$. 
Moreover, since $S$ is a maximizer of $\rho$, then the subset $S'= S \setminus \{u\}$ has density less than or equal to $\rho(S)$. Hence, we have:
\[
\rho(S') = \frac{m_S -(k-1)}{|S| -1} \leq \frac{m_S}{|S|}, 
\]
which is equivalent to $m_S \leq (k-1)|S|$. 
On the other hand, since $k < 2C+1$, there must be some vertex $w \notin S$ which belongs to the attached clique of $v$. Consider the set $S'' =  S \cup \{w\}$.
Its density is also less than or equal to $\rho(S)$, hence we have:
\[
\rho(S'') = \frac{m_S + k}{|S| + 1} \leq \frac{m_S}{|S|}, 
\]
which is equivalent to $m_S \geq k |S|$, a contradiction to the previous inequality. Thus, it must be the case that $k=1$ or $k=2C+1$.

Let us now show that $k = 2C+1$. For the sake of contradiction, assume that $k=1$, i.e., $S$ contains only $v$. 
Consider the set $S'$, obtained from $S$ by including all the $2C$ remaining vertices of the clique. It has density: 
\[
\rho(S') = \frac{m_S +\binom{2C+1}{2}}{|S| + 2C},
\]
which is strictly larger than $m_S/|S|$, provided that $m_S/|S| < C + 1/2$. Let us prove then that $m_S/|S| < C + 1/2$. 

Denote by $v_S = |S \cap V(G)|$ and by $e_S$ the number of edges of $G$ induced by $S$. 
Then, by the previous claim, we have that $|S| = v_S + t \cdot 2C$ and $m_S = e_S + t \binom{C+1}{2}$, for some $t \geq 0$ representing the number of attached cliques that are fully
contained by $S$. Hence, we have that:
\[
\frac{m_S}{|S|} = \frac{e_S + t \binom{C+1}{2}}{v_S + t \cdot 2C} < C+1/2,
\]
where the inequality follows from our assumption on $\rho(G)$, namely $e_S/v_S \leq \rho(G) < C+1/2$. 
Hence, $\rho(S') > \rho(S)$, which contradics $S$ being a maximizer of $\rho$. Therefore $k = 2C+1$. 

We showed that $S$ is a subset of vertices of $G$ and their corresponding $(2C+1)$-cliques.
In particular $|S| = v_S (2C+1) $, and $m_S = e_S + v_S \binom{2C+1}{2}$. Therefore: 
\[
\rho(S) = \frac{e_S + v_S \binom{2C+1}{2}}{v_S(2C+1)} = \frac{e_S}{v_S(2C+1)} + C.
\]
\noindent 
Note that since $e_S/v_S \leq \rho(G)$, then $\rho(S) \leq \rho(G)/(2C+1) + C$. However, it must be the case that $v_S/e_S = \rho(G)$, 
since this value can be attained by picking any densest subgraph in $G$ and then adding all the attached cliques incident to it. Therefore, $\rho(S) = \rho(G)/(2C+1)+C$. 
\end{proof}

\section{Definition of Dynamic Direct-WTERs}
\label{app:ddwter}
\begin{definition}[Dynamic Direct-WTER]
\label{def:DD-WTER}
Let $\mathcal{A}$ be a dynamic graph problem, on a dynamic graph $G_0,G_1,\ldots$, where $G_t := (V,E_t)$ is the graph after the $t^{th}$ edge update. Let us denote $n := |V|$ and $m_t := |E_t|$. 
A Dynamic Direct-WTER to $\mathcal{A}$, is a dynamic algorithm, that: (1) given $G_0$, it outputs in $\tilde{O}(n+m_0)$ time an $\Omega(1)$-expander 
$G^{exp}_{i_0} := (V_{exp},E^{exp}_{i_0})$ with $i_0 := 0$, and (2) given an edge update to $G_{t-1}$, it computes in amortized $\tilde{O}(1)$ time a sequence of edge updates to $G^{exp}_{i_{t-1}}$, 
whose length is $i_t - i_{t-1}$ for some $i_t > i_{t-1}$, resulting in a graph $G^{exp}_{i_t}$, such that for every $t \in \mathbb{N}$ we have the following guarantees:
\begin{itemize}
\item The length of the sequence of updates to $G^{exp}_{t-1}$ is \emph{amortized} $O(1)$, i.e., $i_t = O(t)$. 
\item Every intermediate graph in $G^{exp}_{i_{t-1}}, G^{exp}_{i_{t-1}+1},\ldots,G^{exp}_{i_t}$ is an $\Omega(1)$-expander. In addition, the blowup in the number of vertices and edges is $|V^{exp}| \leq K$ and $|E^{exp}_{i_{t-1}+i}| = {O}(M)$ for some $K := L(n,m_0)$ and $M := M(n,m_t)$. 
\item The solution $\mathcal{A}(G_{t})$ can be computed in $\tilde{O}(1)$ time from the solution $\mathcal{A}(G^{exp}_{i_t})$. 
\end{itemize}
\end{definition}
\noindent 
Observe that if problem $\mathcal{A}$ admits a DD-WTER with linear blowup, i.e. $K = O(n)$ and $M = O(N+m_t)$, then the amortized complexity of the problem on $\Omega(1)$-expanders is equivalent to its amortized complexity on general graphs, 
up to $poly(\log n)$ factors, and assuming $m = \Omega(n)$. We remark that our definition is w.r.t. the amortized complexity and not worst-case complexity of problems. 
The reason for this is, that making our Direct-WTERs work both in \emph{deterministic and worst-case} time per update proved more challenging. Instead, we keep our focus on the deterministic regime, 
and show Direct-WTERs that work in deterministic amortized $\tilde{O}(1)$ time per update.

\section{Explicit Expander Construction}
\label{app:explicit}
In the appendix, we present an explicit, deterministic construction of bipartite $\Omega(1)$-expanders, which builds upon known constructions from the literature. 
For this part only, we will need to introduce the notion of spectral expanders, which are more common in the literature about explicit constructions. 

\paragraph{Graph spectrum.}
Consider a $d$-regular, $n$-vertex graph $G$.
Denote by $\lambda_1 \geq \lambda_2 \geq \cdots \geq \lambda_n$ the spectrum of its adjacency matrix. 
It is known that $\lambda_1 = d$, and if $G$ is not bipartite then $\lambda_n > -d$.
There is also a close relationship between $\lambda_2$ and the edge-expansion of $G$, 
which follows by the following theorem.
\begin{theorem}[Cheeger's inequality]
\[
\frac{d-\lambda_2}{2} \leq h_G \leq \sqrt{2 d ( d- \lambda_2)}
\]
\end{theorem}
In particular, if $\lambda_2$ is sufficiently small, say $\lambda_2 \leq 0.1 d$, then $G$ 
is a $\Omega(d)$-edge expander and by \cref{prop:expansion} also an $\Omega(1)$-expander. 

In our core gadget (\cref{sec:coregadget}), we will utilize a \emph{strongly explicit} construction of expanders by Noga Alon~\cite[Proposition 1.1]{alon2021explicit}. 
This is a deterministic construction that produces simple, $d$-regular graphs on $n$ vertices, 
such that $|\lambda_i| = O( \sqrt{d} )$ for all $i=\in[2,n]$. 
By strongly explicit, we mean not only the ability to construct the entire graph deterministically, 
but also the capability to determine the $i^{th}$ neighbor of vertex $u$ in $\tilde{O}( 1 )$ time, given $i \in [d]$ and $u \in [n]$.
In particular, we can construct the entire graph in $\tilde{O}(nd)$ time. 
For any sufficiently large value of $n$ and $d\geq 3$, this construction is applicable with $n' = n(1+o_n(1))$ vertices, where the $o_n(1)$ term tends to zero as $n$ tends to infinity. 
For the sake of simplicity, we will assume that $n' = n$, 
as this will not impact our final results. 

In fact, we require a strongly explicit construction of a bipartite expander graph since Alon's graphs are inherently non-bipartite. To achieve this, we adapt the construction through a technique known as a \emph{$2$-lift}. This adaptation yields a strongly explicit bipartite graph that ensures a bounded second-largest eigenvalue.

Let's consider a strongly explicit graph $G = (V,E)$, where $|\lambda_i| \leq \lambda$ for all eigenvalues $i\in[2,n]$. We establish two copies of $V$, labeled as $V_1$ and $V_2$. Edges $(u_1, v_2) \in V_1 \times V_2$ are introduced if and only if the corresponding edge $(u, v) \in E$. This operation results in a $d$-regular bipartite graph that retains strong explicitness.

According to~\cite[Lemma 3.1]{bilu2006lifts}, the eigenvalues of this new bipartite graph correspond to those of $G$ and their additive inverses. Notably, the second largest eigenvalue remains bounded by $\lambda$ through this transformation.

\section{Problem Definitions}
\label{app:problems}
In this appendix, we provide the definitions and some background on each problem that appears in the paper. 

\paragraph{Maximum Matching and Bipartite Perfect Matching.} 
In the Maximum Matching problem, the goal is to compute (in the dynamic setting, maintain) the maximum cardinality of any \emph{matching}, that is, a
set of edges $M \subseteq E$ in which no two edges intersect. 
Relevant background for this problem in the static setting, including background on its expander-case and average-case complexities, can be found in~\cite{abboud2023worst}. 
Relevant background for this problem in the dynamic setting, including background on its worst-case lower bounds and the bounded-degree case, 
can be found in~\cite{henzinger2015unifying}. 

In the Bipartite Perfect Matching problem, the graph is bipartite, and both of its parts have the same size $n$. The goal is to decide whether the maximum matching has cardinality $n$ or not. A recent algorithm for this problem solves it in near-linear time~\cite{chen2022maximum}, but it remains hard in the dynamic setting. 
Conditional lower bounds and background for this problem can be found in~\cite{abboud2014popular} and~\cite{henzinger2015unifying}. 

\paragraph{$k$-Clique Detection, $k$-Clique Counting, and Max-Clique.}
In the $k$-Clique Detection problem, the goal is to decide whether a graph contains a $k$-clique as a subgraph. 
In the $k$-Clique Counting problem, the goal is to count the number of copies of a $k$-cliques in a graph. 
Both of these problems are fundamental ``hard'' problems in fine-grained complexity. See for reference~\cite{vassilevska2009efficient} and~\cite{ABW18}.
In the Max-Clique problem, the goal is to compute that largest $k$ such that $G$ 
contains a $k$-clique as a subgraph. 
This is one of the most well-known NP-hard problems. See~\cite{xiao2017exact} for the exponential complexity of this problem (presented for the equivalent MIS problems). 
I

\paragraph{Densest Subgraph.}
In the Densest Subgraph problem, we define the \emph{density} of a subset $\emptyset \neq S \subseteq V$ as $\rho(S) := e(S)/|S|$, 
where $e(S)$ is the number of edges in the subgraph induced by $S$. The goal is to find the maximum density of any subset in the graph. 
The goal is to compute the maximum {density} of any induced subgraph in a graph. 
In the static setting, this problem can be solved via $O(\log n )$ max-flow computations,
~\cite{goldberg1984finding}, which results in a near-linear time algorithm~\cite{chen2022maximum}. 
However, it remains hard in the dynamic setting. See~\cite{henzinger2015unifying} for some background and a lower bound. 
In the Densest Subgraph ($m \geq \alpha n$) problem, we restrict the input graph to have at least $\alpha n$ edges. 

\paragraph{$H$-Subgraph Detection for $H$ without pendant vertices.}
See~\cite{abboud2023worst} for the definition and background for this problem. 

\paragraph{Max-Cut.}
In this problem, the goal is to compute $\max_{S \subseteq V} e(S,V \setminus S)$.
This problem is one of the most well-known NP-hard problems. The fastest exponential-time algorithm for it is by Williams~\cite{williams2005new}. 

\paragraph{$st$-SP.}
Given a dynamic graph with two distinguished vertices $s$ and $t$, maintain the length of a shortest path between them in the graph. 
See~\cite{abboud2014popular} and~\cite{henzinger2015unifying} for background on this problem. 

\paragraph{Online Matrix-vector Multiplication Problem (OMv).} [\cite{henzinger2015unifying,williams2007matrix}]
In this problem, an algorithm is given an $n$-by-$n$ boolean matrix $M$, preprocesses is in polynomial time, and then it receives a sequence of $n$ vectors 
$v_1,v_2,\ldots,v_n \in \{0,1\}^n$ one by one. The goal is to compute the boolean product $Mv_i$ for every $i$ before seeing the next vector in the sequence. 
See~\cite{henzinger2015unifying} for background on this problem, where they also propose the OMv Conjecture:
\begin{conjecture}[OMv Conjecture~\cite{henzinger2015unifying}]
\label{conj:omv}
For any constant $\eps >0$, there is no $O(n^{3-\eps})$-time randomized algorithm that solves OMv with an error probability of at most $1/3$.
\end{conjecture}
\noindent
In addition, the authors also prove that the OMv problem is essentially equivalent to the OuMv problem, where instead of a sequence of vectors, the algorithm receives 
a sequence of pairs of vectors $(u_1,v_1),\ldots,(u_n,v_n)$, and the goal is to compute the boolean product $u_i^T M v_i$ for every $i$.

\end{document}